\documentclass[a4paper,UKenglish,cleveref, autoref, thm-restate]{lipics-v2021}

\usepackage{graphicx}
\usepackage{xcolor}
\definecolor{babyblueeyes}{rgb}{0.63, 0.79, 0.95}
\definecolor{dbabyblueeyes}{rgb}{0.50, 0.69, 0.85}
\definecolor{brightube}{rgb}{0.82, 0.62, 0.91}
\definecolor{salmonpink}{rgb}{1.0, 0.57, 0.57}
\definecolor{thistle}{rgb}{0.85, 0.68, 0.87}
\definecolor{apricot}{rgb}{0.98, 0.81, 0.69}
\usepackage{tikz}
\usetikzlibrary{arrows.meta,bending,positioning,decorations.text}
\usetikzlibrary{calc}

\usetikzlibrary{decorations.pathreplacing,calligraphy}
\usetikzlibrary{tikzmark}
\usepackage{colortbl}
\usepackage{booktabs, multirow, tabularx}
\definecolor{green}{rgb}{0.1,0.1,0.1}
\newcommand{\cone}{\cellcolor{pink}}  
\newcommand{\ctwo}{\cellcolor{babyblueeyes}}  
\newcommand{\cthree}{\cellcolor{thistle}}  
\newcommand{\cfour}{\cellcolor{apricot}}  
\usepackage[ruled]{algorithm2e} 

\SetAlFnt{\small}
\SetAlCapFnt{\small}
\SetAlCapNameFnt{\small}
\SetAlCapHSkip{0pt}
\IncMargin{-\parindent}
\newcommand{\appsymb}{$\bigstar$}

\newcommand{\pavsc}{\textsc{pavsc}}

\title{A Lower Bound for Local Search Proportional Approval Voting} 

\author{Sonja Kraiczy}{Department of Computer Science, University of Oxford, UK}{Sonja.Kraiczy@cs.ox.ac.uk}{https://orcid.org/0009-0005-0670-6201}{The author is grateful to be supported by an EPSRC studentship.}

\author{Edith Elkind}{Department of Computer Science, University of Oxford, UK and Alan Turing Institute, UK}{Edith.Elkind@ox.ac.uk}{https://orcid.org/0000-0001-6718-3436}{This work is supported by AI Programme of The Alan Turing Institute and EPSRC grant EP/X038548/1.}
\authorrunning{S. Kraiczy and E. Elkind} 
\Copyright{Sonja Kraiczy and Edith Elkind} 

\ccsdesc[100]{Theory of computation $\rightarrow$ Theory and algorithms for application domains $\rightarrow$ Algorithmic game theory and mechanism design $\rightarrow$ Exact and approximate computation of equilibria} 

\keywords{Computational Social Choice, Committee Elections, Local Search, Fairness} 

\category{} 

\relatedversion{} 

\nolinenumbers 

\EventEditors{Timothy Chan, Johannes Fischer, John Iacono, and Grzegorz Herman}
\EventNoEds{4}
\EventLongTitle{32nd Annual European Symposium on Algorithms (ESA 2024)}
\EventShortTitle{ESA 2024}
\EventAcronym{ESA}
\EventYear{2024}
\EventDate{September 2--4, 2024}
\EventLocation{Royal Holloway, London, United Kingdom}
\EventLogo{}
\SeriesVolume{308}
\ArticleNo{27}

\begin{document}

\maketitle

\begin{abstract}
Selecting $k$ out of $m$ items based on the preferences of
$n$ heterogeneous agents is a widely studied problem
in algorithmic game theory.
If agents have approval preferences over individual items and harmonic utility
functions over bundles---an agent receives $\sum_{j=1}^t\frac{1}{j}$ utility 
if $t$ of her approved items are selected---then welfare optimisation 
is captured by a voting rule known as Proportional Approval Voting (PAV). PAV also satisfies demanding fairness axioms. However, finding a winning set of items under PAV is NP-hard. In search of a tractable method with strong fairness guarantees, a bounded local search version of PAV was proposed \cite{AE+18}. It proceeds by starting with an arbitrary size-$k$ set $W$ and, at each step, checking if there is a pair of candidates $a\in W$, $b\not\in W$ such that swapping $a$ and $b$ increases the total welfare by at least $\varepsilon$; if yes, it performs the swap. Aziz et al.~show that setting $\varepsilon=\frac{n}{k^2}$ ensures both the desired fairness guarantees and polynomial running time. However, they leave it open whether the algorithm converges in polynomial time if $\varepsilon$ is very small (in particular, if we do not stop until there are no welfare-improving swaps).
We resolve this open question, by showing that if $\varepsilon$ can be arbitrarily small, the running time of this algorithm may be super-polynomial. Specifically, we prove a lower bound of~$\Omega(k^{\log k})$ if improvements are chosen lexicographically. To complement our lower bound, we provide an empirical comparison of two variants of local search---better-response and best-response---on several real-life data sets and a variety of synthetic data sets. Our experiments indicate that, empirically, better response exhibits faster running time than best response.
\end{abstract}
\newpage
\section{Introduction}
\label{sec:intro}

We study the collective decision problem where the goal is to select $k$ out of $m$ items (candidates) based on the preferences of $n$ agents (voters).
This problem (or, more precisely, its generalisation to the setting where different items
may have different costs and there is a budget constraint)
is known as the {\em combinatorial public project} problem in the algorithmic mechanism design literature, where the focus in on optimisation of the utilitarian welfare subject to strategyproofness constraints \cite{buchfuhrer2010computation,dughmi2011truthful}. 
In the computational social choice literature, it is known as the the 
{\em multiwinner voting} problem, and an important concern is proportionality, 
i.e., ensuring that large groups of voters with similar preferences
are fairly represented by the elected candidates~\cite{lackner2023multi}.

Proportional Approval Voting (PAV)~\cite{J16}, which belongs to the class of 
Thiele's methods~\cite{AB+17}, is a multiwinner voting rule that can be viewed both from the welfare maximisation perspective and from the 
proportionality perspective. Given an election with a set of voters $N$, 
a set of candidates $C$, a target committee size $k$, and the voters' approval ballots 
$(A_v)_{v\in N}$, where $A_v\subseteq C$ for all $v\in N$, it computes the {\em PAV score} of a committee $W\subseteq C$ as $\pavsc(W)=\sum_{v\in N}\sum_{j=1}^{|A_v\cap W|}\frac{1}{j}$, 
and outputs all size-$k$ committees with the maximum PAV score.
This rule admits a utilitarian interpretation: agent $v$ with ballot $A_v$ is assumed to compute the utility of bundle $W$ as $\sum_{j=1}^{|A_v\cap W|}\frac{1}{j}$, and the rule maximises the utilitarian welfare. Another interpretation makes no assumptions 
about agents' values for bundles, but aims to achieve proportionality. Intuitively, proportionality means that an $\alpha$-proportion of the population that jointly approves at least $\alpha k$ items should be represented by an $\alpha$-fraction of the $k$ selected items. There are several ways to formalise
this intuition, including justified representation axioms, such as PJR and EJR~\cite{AB+17,SE+17}, or the notion of proportionality degree~\cite{skowron2021proportionality}, and PAV is among the very few voting rules
that satisfy EJR and have optimal proportionality degree.
Unfortunately, however, PAV is known to be NP-hard to compute~\cite{skowron2016finding,AG+15}.
\begin{algorithm}
\caption{$\varepsilon$-local search PAV ($\varepsilon$-ls-PAV)}\label{elspav}
\KwData{$\varepsilon>0$, arbitrary committee $W_\text{init}$ of size $k$, voters' approval ballots $(A_v)_{v\in N}$}
\KwResult{$W$ of size $k$}
$\mathit{W} \gets W_\text{init}$\;
\While{$\exists b\notin W, a \in W \textrm{ such that } \Delta(W,a,b) \ge \varepsilon$}
  {$W\gets (W \cup \{b\})\setminus\{a\}$} 
\Return $W$
\end{algorithm}

Since PAV has excellent proportionality properties, there has been a lot of interest in identifying tractable variants of this rule. Two natural approaches to explore in this context are greedy sequential optimisation and local search. The former is a special case of the greedy algorithm for submodular function maximisation, and approximates the social welfare by a factor of $(1-\frac{1}{e})$~\cite{nemhauser1978analysis}. 
Unfortunately, high PAV score does not imply good proportionality guarantees~\cite{AB+17}, so approximation algorithms do not appear to be very useful from the fairness perspective.

In contrast, the local search approach turns out to be well-suited for identifying fair outcomes. 
The reason for this is that proofs of proportionality guarantees for PAV 
use local swap arguments: they show that if a committee $W$ is proportional, there is no pair of candidates $a\in W, b\notin W$ such that swapping them increases the PAV score, i.e., the quantity
$$
\Delta(W,a,b)=\pavsc((W\setminus\{a\})\cup\{b\})-\pavsc(W)
$$ 
is positive. The local search algorithm that
starts with an arbitrary committee and performs PAV score-improving swaps is therefore a natural heuristic for PAV. It was first introduced and studied by Aziz et al.~\cite{AE+18}. However, Aziz et al.~were unable 
to show that local search converges after polynomially many 
iterations, as a single iteration may only increase the PAV score by a tiny amount.
To overcome this issue, they considered a parameterised version of local search, 
which only performs swaps if they improve the PAV score by at least $\varepsilon$
(see Algorithm~\ref{elspav} for the pseudocode).
They observed that if $\varepsilon$ is sufficiently small, the algorithm, which we will call
$\varepsilon$-ls-PAV, preserves
the proportionality guarantees of PAV, and if $\varepsilon$ is sufficiently large, it converges
in polynomial time; setting $\varepsilon=\frac{n}{k^2}$ achieves both of these objectives.
However, Aziz et al.~left it as an open problem if $\varepsilon$-ls-PAV converges in polynomial time
for {\em all} values of $\varepsilon>0$. The following conjecture is implicit in their work:
\begin{conjecture}[left open by \cite{AE+18}]\label{conj}
For small $\varepsilon$, $\varepsilon$-ls-PAV may make a super-polynomial number of swaps.
\end{conjecture}

\begin{figure}[h]
\begin{subfigure}{0.48\linewidth}
\centering
\begin{tikzpicture}[scale=0.8]

  \foreach \x in {-2,-1.5,...,2.5}
    \foreach \y in {0,0.5,...,2.5}
      {
        \fill (\x-1.75,\y-0.75) [color=pink] circle (0.15cm);
      }
\foreach \x in {3.5,4,...,5}
    \foreach \y in {0,0.5,...,2.5}
      {
        \fill (\x-1.75,\y-0.75) [color=babyblueeyes] circle (0.15cm);
      }

        \fill (0,0) node[minimum height=2cm,minimum width=6cm,rounded corners=0.25cm] {};
        \fill (5,0) node[minimum height=2cm,minimum width=3cm] {};

        \draw [line width=1mm, decorate, decoration = {calligraphic brace, mirror, amplitude=5mm}] (-4,-1) --  (1,-1);
        \draw [line width=1mm, decorate, decoration = {calligraphic brace, mirror, amplitude=3mm}] (1.5,-1) --  (3.5,-1);
        \node at (-1.5,-2) {60 voters approve};

        \path[fill=salmonpink] (-2.5,-2.5) circle (0.25cm) node[text=black] {$c_1$};
        \path[fill=salmonpink] (-1.5,-2.5) circle (0.25cm) node[text=black] {$c_2$};
        \path[fill=salmonpink] (-0.5,-2.5) circle (0.25cm) node[text=black] {$c_3$};
        \path[fill=dbabyblueeyes] (2,-2.5) circle (0.25cm) node[text=black] {$c_4$};
        \path[fill=dbabyblueeyes] (3,-2.5) circle (0.25cm) node[text=black] {$c_5$};

        \node at (2.5,-2) {30 approve};
\end{tikzpicture}
\subcaption{60 voters approve $c_1,c_2,c_3$\\ and 30 voters approve $c_4,c_5$\\}\label{sixtythirty}
\end{subfigure}
\begin{subfigure}{0.48\linewidth}

\begin{tikzpicture}[scale=0.8]

  \foreach \x in {-2,-1.5,...,2.5}
    \foreach \y in {0,0.5,...,2.5}
      {
        \fill (\x-2.75,\y-0.75) [color=pink] circle (0.15cm);
      }
\foreach \x in {3.5,4,...,5}
    \foreach \y in {0,0.5,...,2.5}
      {
        \fill (\x-2.75,\y-0.75) [color=babyblueeyes] circle (0.15cm);
      }

        \fill (-0.25,1.75) [color=white] circle (0.16cm);
         \fill (-0.25,1.25) [color=white] circle (0.16cm);
        \fill (-0.75,1.75) [color=white] circle (0.16cm);
        \fill (2.25,1.75) [color=white] circle (0.16cm);
        \fill (2.25,1.25) [color=white] circle (0.16cm);
        \fill (-0.75,1.25) [color=white] circle (0.16cm);

   \foreach \y in {0,0.5,...,2.5}
      {
        \fill (3.25,\y-0.75) [color=thistle] circle (0.15cm);
      }
      
        \fill (-1,0) node[minimum height=2cm,minimum width=6cm,rounded corners=0.25cm] {};
        \fill (4,0) node[minimum height=2cm,minimum width=3cm] {};

        \draw [line width=1mm, decorate, decoration = {calligraphic brace, mirror, amplitude=5mm}] (-5,-1) --  (0,-1);
        \draw [line width=1mm, decorate, decoration = {calligraphic brace, mirror, amplitude=3mm}] (0.5,-1) --  (2.5,-1);
        \draw [line width=0.5mm, decorate, decoration = {calligraphic brace, mirror, amplitude=1mm}] (3,-1) --  (3.5,-1);

        \node at (-2.5,-2) {56 voters approve};
        \path[fill=salmonpink] (-3.5,-2.5) circle (0.25cm) node[text=black] {$c_1$};
        \path[fill=salmonpink] (-2.5,-2.5) circle (0.25cm) node[text=black] {$c_2$};
        \path[fill=salmonpink] (-1.5,-2.5) circle (0.25cm) node[text=black] {$c_3$};
        \path[fill=dbabyblueeyes] (1,-2.5) circle (0.25cm) node[text=black] {$c_4$};
        \path[fill=dbabyblueeyes] (2,-2.5) circle (0.25cm) node[text=black] {$c_5$};
        \node at (1.5,-2) {28 approve};
        \node at (3.25,-2) {6};

\end{tikzpicture}
\subcaption{6 voters leave their groups and vote for themselves. 56 voters approve $c_1,c_2,c_3$ and 28 voters approve $c_4,c_5$}\label{noise}
\end{subfigure}
\caption{For $k=3$ and initial committee $W=\{c_1,c_2,c_3\}$ , $\frac{n}{k^2}$-ls-PAV would replace a member of $W$ with one of $c_4,c_5$ in the instance in Figure~\ref{sixtythirty}, but not in the instance in Figure~\ref{noise}.}
\end{figure}
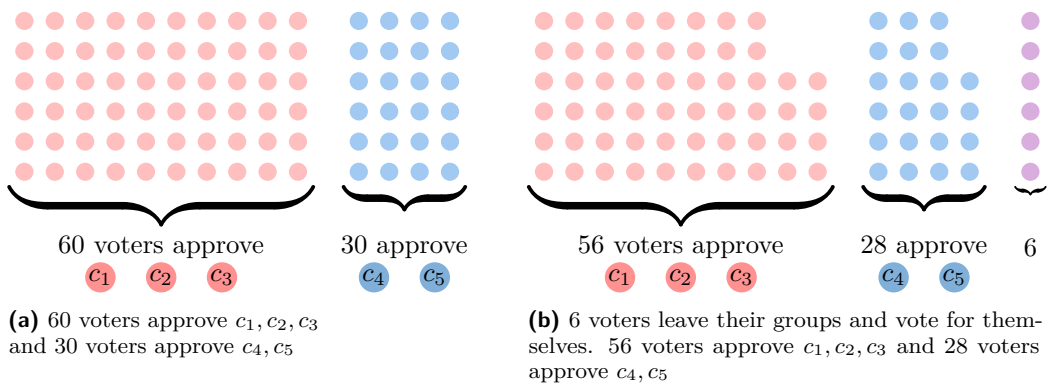
Now, while setting $\varepsilon=\frac{n}{k^2}$ preserves the worst-case proportionality guarantees, the faster running time 
comes at a cost: for large values of $n$ and small values of $k$ the algorithm
may get ``lazy'' and choose not to perform a swap even if it would result
in a much fairer outcome. Indeed, consider an instance with $90$ voters, $60$ of which vote for $\{c_1, c_2, c_3\}$ and $30$ vote for $\{c_4, c_5\}$ (see Figure~\ref{sixtythirty}).
For a committee size of $k=3$, a fair outcome gives two representatives to the group of $60$ and one to the group of $30$. Indeed, one can check that the majoritarian committee $W=\{c_1,c_2,c_3\}$ is not locally optimal, in the sense that 
$\frac{n}{k^2}$-ls-PAV will replace one of the candidates in $W$ with $c_4$ or $c_5$.
However, if the voters' preferences are a bit more diffuse, this is no longer the case.
To see this, consider Figure~\ref{noise}, where the groups of size $60$ and $30$ 
lose four and two members, respectively, and the six breakaway voters 
become candidates and vote for themselves.
For this instance, $W=\{c_1, c_2, c_3\}$ is an output of $10$-ls-PAV, because the PAV scores
of $W$ and $(W\setminus\{c_3\})\cup\{c_4\}$ are, respectively, $56\cdot(1+1/2+1/3)$ and 
$56\cdot(1+1/2)+28$, and $28-\frac{56}{3} < \frac{n}{k^2}=10$. On the other hand, 
for the ``vanilla'' version of ls-PAV, which does not stop until there are no
PAV-score improving swaps, the ``unfair'' committee $W$ is not among the outputs of local search on the modified instance.

There are other reasons to favour ``vanilla'' ls-PAV over $\frac{n}{k^2}$-ls-PAV. For instance, the former rule is easier to explain to voters, who may be disappointed that the latter rule stops despite the availability of score-improving swaps. Also, ``vanilla'' ls-PAV is more decisive---it is easy to see that for $\varepsilon<\varepsilon'$
each output of $\varepsilon$-ls-PAV is an output of $\varepsilon'$-ls-PAV, 
but the converse is not necessarily true---and decisiveness is viewed as a desirable property in the social choice literature. Thus, it is important to understand whether the conjecture of \cite{AE+18} is true, i.e., whether $\varepsilon$-ls-PAV converges in polynomial time even for small values of $\varepsilon$.

\smallskip
Motivated by these considerations, in this work we investigate Conjecture~\ref{conj} 
and resolve it in the positive. Specifically, for adversarially chosen swaps 
we show a lower bound of $\Omega(k^{\log k})$ via a subtle combinatorial construction. We then extend our result to a natural pivoting rule, which chooses swaps lexicographically. 
While our result does not rule out the possibility that the 
outputs of $\varepsilon$-ls-PAV can be computed by a different polynomial-time algorithm (the complexity of this problem remains an interesting open question), it justifies the choice of $\varepsilon=\frac{n}{k^2}$ in the work of Aziz et al.~\cite{AE+18}.

Importantly, while our lower bound holds for better response with the lexicographic pivoting
rule, it does not extend to another natural pivoting rule: choosing the swap
that offers the maximum increase in the PAV score, i.e., best-response dynamics.
Indeed, it remains an open problem to determine if the best-response variant of ls-PAV
may take superpolynomially many steps. Motivated by this question, in the extended version
we provide an empirical comparison between lexicographic better response and best response,
using several real-life and synthetic datasets. We measure the performance of each algorithm on a given instance as the number of candidate swaps it needs to consider before termination (this is a useful proxy for running time as long as we do not have access to parallel processing hardware). Interestingly, on the datasets we investigate, better response considers fewer swaps than best response. Hence, while our theoretical worst-case results seem to favour best response over better response, the empirical results paint the opposite picture. 
\subsection{Roadmap of our Approach}
We study the runtime of $\varepsilon$-ls PAV with the threshold $\varepsilon$ set to a very small positive value (which we will denote by $0^+$). 
Specifically, for each $k\ge 0$ we construct a multiwinner election $E^k$ with target committee size $k$ and the following properties: (1) the number of voters in $E^k$ is polynomial in $k$; (2) on $E^k$, $0^+$-ls-PAV may require $\Omega(k^{\log k})$ steps until convergence. 
Our argument is closely connected to the following simple-to-formulate number-theoretic question:
 \begin{align}
 \text{minimise } &\sum_{i=1}^{k}\frac{w_i}{i}\qquad\text{subject to} \label{eq:sum}\\
 &\sum_{i=1}^{k}\frac{w_i}{i}>0, \qquad
 \text{ where } w_i\in\mathbb{Z} \text{ and } w_i=\mathrm{poly}(k)\label{eq}
 \end{align}
If $w_i$ are not required to be polynomial in $k$, it is not difficult to see that the minimum of the sum in~\eqref{eq:sum} is the inverse of the least common multiple 
of $\{1,2,\ldots, k\}$.
However, in our construction the number $\sum_{i=1}^k|w_i|$ corresponds to the size of the voter set. Hence, if we want the size of $E^k$ to be polynomial in $k$,
the $w_i$ need to be polynomial in $k$ as well.
At the heart of our construction is an instance of size polynomial in $k$ that corresponds to a value of $\frac{1}{k^{\log k}}$ for the above objective.
We use it to build an election with a carefully crafted combinatorial structure on which $0^+$-ls-PAV is forced to repeatedly undo previously achieved progress.
We note that if $w_i\in\{-1,1\}$, then the best known upper bound is due to \cite{BM+18} and matches our construction in Lemma $4.3$. 

\smallskip
As a warm-up, in Section \ref{sec:warmup}, we first prove a simpler result, which illustrates one of the main ideas behind the more complex
construction. We study $\frac{n}{k^2}$-ls-PAV where swaps may be chosen adversarially, 
and show that its worst-case runtime is $\widetilde{\Theta}(k^2)$.
 In Section~\ref{sec:level1} we first construct an election together with a committee where 
 a swap increases the PAV score just by $\Theta(\frac{1}{k^{\log k}})$.
We use this as a building block to construct two further levels of complexity, which give us our desired instance (Sections~\ref{sec:level2} and~\ref{sec:level3}). This instance is used to show our main result for the adversarial setting (Section~\ref{sec:abr}):
In Theorem~\ref{thm:zerolow}, we exhibit a sequence of swaps of length super-polynomial in $k$ that may be performed by $0^+$-ls-PAV. 
Finally, in Section~\ref{sec:nabr} we show that lexicographic better response executes a subsequence of our sequence from Theorem~\ref{thm:zerolow}, and the length of this subsequence is still superpolynomial in $k$.

\subsection{Related Work}

Initially, proportionality in multiwinner 
committee voting was considered from an axiomatic perspective: there is a spectrum of justified representation
axioms, ranging from Justified Representation (JR) (which is rather mild and easy to satisfy) to more demanding axioms such as Proportional, Extended and Full Justified Representation (PJR,  EJR and FJR)~\cite{AB+17,SE+17,PPS21}; as well as EJR+ \cite{BP23};
 $\frac{n}{k^2}$-ls-PAV is among the very few polynomial-time computable voting rules that satisfy EJR.
Subsequently, Skowron~\cite{skowron2021proportionality} pursued a qualitative approach, and put forward the notion of the proportionality degree, which formalises to what degree $\beta$ it holds that an $\alpha$ proportion of the population proposing at least $\alpha k$ candidates will be represented (so, ideally 
$\beta\sim\alpha$ and the group is represented by roughly $\beta k$ candidates on average). PAV exhibits excellent performance according to this measure: for PAV the ratio $\frac{\beta}{\alpha}$ approaches 1. For $k\leq 200$ Skowron shows that Sequential PAV has a proportionality degree of at least $0.7$, but for larger $k$ the proportionality of Sequential PAV remains open.
Other voting rules with good proportionality guarantees are Sequential Phragm{\'e}n's rule (which satisfies PJR) \cite{J16,BF+17} and the Method of Equal Shares (which satisfies the stronger EJR axiom) \cite{PS20}. Unlike PAV, both of these rules are formulated in terms of voters sharing the `load' incurred by the candidates in the committee, and have a proportionality degree of $0.5$. 


\section{Preliminaries}
For $n\in \mathbb{N}$, we write $[n]=\{1,\ldots, n\}$.
An {\em approval election} is a $4$-tuple
$E = (N, C, (A_v)_{v\in N}, k)$, where
$N=[n]$ is a set of {\em voters}, $C$ is  a set of {\em candidates}, $|C|=m$, $A_v\subseteq C$ is the {\em ballot} of voter $v\in N$, and $k\in [m]$ is the {\em target committee size}. Subsets of $C$ (not necessarily of size $k$) are called {\em committees}. 

We define the {\em PAV satisfaction} of a set of voters $V$ from a committee $W$ 
as $\pavsc_V(W)=\sum_{v\in V}\sum_{j=1}^{|A_v\cap W|}\frac{1}{j}$; if $V$ is a singleton, i.e., 
$V=\{v\}$, we omit the braces and write $\pavsc_v$ instead of $\pavsc_{\{v\}}$.
Given a committee $W$, a pair of candidates $b\notin W$, $a\in W$, 
and a set of voters $V$, we denote by
$\Delta_V(W, a, b)$ the change in the PAV-satisfaction of $V$ that is caused by swapping $a$ with $b$:
$$
\Delta_V(W,a,b)=\pavsc_V(W\cup\{b\}\setminus\{a\})-\pavsc_V(W).
$$
For readability, we omit $V$ from the notation when $V=N$, i.e., we write 
$\Delta(W,a,b) := \Delta_N(W,a,b)$.

The ``vanilla'' local search algorithm, which swaps $a\in W$ with $b\not\in W$ as long as 
$\Delta(W, a, b)>0$, can be described as $\varepsilon$-ls-PAV for $\varepsilon\le \min\{\Delta(W, a, b): W\subseteq C, |W|=k, b\not\in W, a\in W, \Delta(W, a, b)>0\}$. It can be shown that this condition can be satisfied by setting $\varepsilon=\frac{1}{\text{lcm}([k])}$, where 
for each $S\subset \mathbb N$ we denote by $\text{lcm}(S)$ the least common multiple of the integers in $S$.
In what follows, we denote this value of $\varepsilon$ by $0^+$.

\smallskip
 Given a committee $W$, we say that a sequence of swaps 
 $$
 \textbf{X}=(a_1,b_1),(a_2,b_2),\ldots, (a_s,b_s)
 $$ 
 is \textit{valid} if for each $t\in [s]$ the committee $W_t=(W_{t-1}\cup\{b_{t}\})\setminus 
 \{a_{t}\}$ (where $W_0:=W$) satisfies 
 $a_{t}\in W_{t-1}$, $b_{t} \notin W_{t-1}$.
The {\em length} of a sequence $\textbf{X}$, denoted by $|\textbf{X}|$, is the number of pairs in $\textbf X$. We define the {\em inverse} (sequence) of $\textbf{X}$ as $\textbf{X}^{-1}=(b_s,a_s),(b_{s-1},a_{s-1}),\ldots, (b_1,a_1)$.
Given two finite sequences of swaps \textbf{X} and \textbf{Y}, we define their {\em concatenation} $\textbf{X}\oplus \textbf{Y}$ as the sequence with prefix \textbf{X} followed by suffix \textbf{Y}.
For our proofs, it will be useful to have an arbitrarily large pool of `dummy' candidates. We therefore define $D_k=\{d_1,\ldots, d_k\}$ so that $D_{k+1}=D_k\cup \{d_{k+1}\}$, $D_0=\varnothing$ and $D=\cup_{k=1}^{\infty}\{D_k\}$.

We omit some proofs due to space constraints; the respective claims are marked with $\bigstar$. All missing proofs and the simulation results appear in the extended version of the paper. 

\section{Warm-up: Lower Bound for $\frac{n}{k^2}$-ls-PAV}\label{sec:warmup}
To showcase the ideas behind our main lower bound (Theorem~\ref{thm:nabr}), we start by presenting a lower bound on $\frac{n}{k^2}$-ls-PAV in the adversarial setting. Specifically, for each $n$-voter election and committee size $k$, we consider a directed graph whose vertices are committees, and whose edges are pairs $(W,W')$ such that $W'$ can be obtained from $W$ via a swap and the PAV score of $W'$ is at least $\frac{n}{k^2}$ higher than that of $W$. We exhibit an election for which this graph contains a path of length $\Omega(k^2)$. We call this setting adversarial, because these swaps may be suggested by an adversary whose aim is to maximise the number of iterations. 

This result establishes
that the upper bound on the number of iterations of the algorithm of \cite{AE+18} is tight up to a $\log k$ factor. 
Indeed, since the maximum PAV score
of a size-$k$ committee is $O(n\log k)$, it follows that $\frac{n}{k^2}$-ls-PAV
converges in at most $O(k^2\log k)$ iterations.
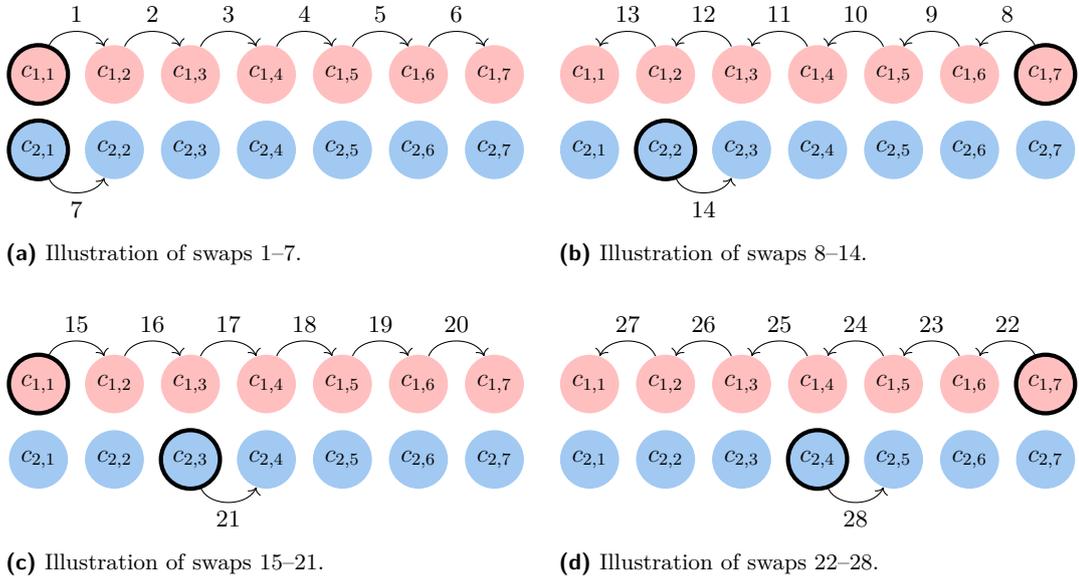
\begin{figure}[h]
\centering
\begin{subfigure}{0.48\linewidth}

\begin{tikzpicture}

 \foreach \x in {1,2,...,7}
      {
        \node[fill=pink,shape=circle, scale=0.8, font=\Large] (node\x) at (\x-3.75,0) { $c_{1,\x}$};
        \node[fill=babyblueeyes,shape=circle, scale=0.8, font=\Large] (node2\x) at (\x-3.75,-1) { $c_{2,\x}$};

      }
 \foreach \x [count= \y from 2] in {1,2,...,6}
      {
               \draw (node\x) edge [->, bend left=70] node[above, font=\small] {\x} (node\y);
        }
                       \draw (node21) edge [->, bend right=70] node[below, font=\small] {7} (node22);

                \node[draw, ultra thick, fill=pink,shape=circle, scale=0.8, font=\Large] (node1) at (-2.75,0) { $c_{1,1}$};
                \node[draw, ultra thick, fill=babyblueeyes,shape=circle, scale=0.8, font=\Large] (node1) at (-2.75,-1) { $c_{2,1}$};

\end{tikzpicture}\caption{Illustration of swaps 1--7.}
\end{subfigure}
\hfill
\begin{subfigure}{0.48\linewidth}
\begin{tikzpicture}

 \foreach \x in {1,2,...,7}
      {
     \node[fill=pink,shape=circle, scale=0.8, font=\Large] (node\x) at (\x-3.75,0) { $c_{1,\x}$};
    \node[fill=babyblueeyes,shape=circle, scale=0.8, font=\Large] (node2\x) at (\x-3.75,-1) { $c_{2,\x}$};

      }

    \foreach \x in {1,2,...,6}
      {
                \pgfmathsetmacro{\y}{int(\x+1)};
                \pgfmathsetmacro{\l}{int(13-\x+1)}

               \draw (node\y) edge [->, bend right=70] node[above, font=\small] {\l} (node\x);
        }
    \draw (node22) edge [->, bend right=70] node[below, font=\small] {14} (node23);

    \node[draw, ultra thick, fill=pink,shape=circle, scale=0.8, font=\Large] (node1) at (3.25,0) { $c_{1,7}$};
    \node[draw, ultra thick, fill=babyblueeyes,shape=circle, scale=0.8, font=\Large] (node1) at (-1.75,-1) { $c_{2,2}$};

\end{tikzpicture}\caption{Illustration of swaps 8--14.}
\end{subfigure}\vspace{0.5cm}
\begin{subfigure}{0.48\linewidth}

\begin{tikzpicture}

 \foreach \x in {1,2,...,7}
      {
     \node[fill=pink,shape=circle, scale=0.8, font=\Large] (node\x) at (\x-3.75,0) { $c_{1,\x}$};
    \node[fill=babyblueeyes,shape=circle, scale=0.8, font=\Large] (node2\x) at (\x-3.75,-1) { $c_{2,\x}$};

      }

    \foreach \x in {1,2,...,6}
      {
                \pgfmathsetmacro{\y}{int(\x+1)};
                \pgfmathsetmacro{\l}{int(14+\x)}

               \draw (node\x) edge [->, bend left=70] node[above, font=\small] {\l} (node\y);
        }
    \draw (node23) edge [->, bend right=70] node[below, font=\small] {21} (node24);

    \node[draw, ultra thick, fill=pink,shape=circle, scale=0.8, font=\Large] (node1) at (-2.75,0) { $c_{1,1}$};
    \node[draw, ultra thick, fill=babyblueeyes,shape=circle, scale=0.8, font=\Large] (node1) at (-0.75,-1) { $c_{2,3}$};

\end{tikzpicture}\caption{Illustration of swaps 15--21.}
\end{subfigure}\hfill
\begin{subfigure}{0.48\linewidth}

\begin{tikzpicture}

 \foreach \x in {1,2,...,7}
      {
     \node[fill=pink,shape=circle, scale=0.8, font=\Large] (node\x) at (\x-3.75,0) { $c_{1,\x}$};
    \node[fill=babyblueeyes,shape=circle, scale=0.8, font=\Large] (node2\x) at (\x-3.75,-1) { $c_{2,\x}$};

      }

    \foreach \x in {1,2,...,6}
      {
                \pgfmathsetmacro{\y}{int(\x+1)};
                \pgfmathsetmacro{\l}{int(27-\x+1)}

               \draw (node\y) edge [->, bend right=70] node[above, font=\small] {\l} (node\x);
        }
    \draw (node24) edge [->, bend right=70] node[below, font=\small] {28} (node25);

    \node[draw, ultra thick, fill=pink,shape=circle, scale=0.8, font=\Large] (node1) at (3.25,0) { $c_{1,7}$};
    \node[draw, ultra thick, fill=babyblueeyes,shape=circle, scale=0.8, font=\Large] (node1) at (0.25,-1) { $c_{2,4}$};

\end{tikzpicture}\caption{Illustration of swaps 22--28.}
\end{subfigure}
\caption{Highlighted candidates are in the committee; arrows from $a$ to $b$ labelled with $t$ indicate that $a$ is replaced by $b$ in iteration $t$. Observe that pink candidates are only replaced by pink candidates; similarly, blue candidates are only replaced by blue candidates. This illustrates a swap sequence similar to that in the proof of Theorem~\ref{thm:threshold}, with the exception that, to make the figure visually appealing, we display an equal number of blue and pink candidates. }
\end{figure}
\begin{theorem}\label{thm:threshold} 

For every $k\geq 1$ there exists a committee election with $\mathrm{poly}(k)$ voters, 
a committee $W_0$, $|W_0|=k$, and a sequence of $\Omega(k^2)$ swaps starting from $W_0$ such that each swap in this sequence strictly increases the PAV score by at least $\frac{n}{k^2}$.
\end{theorem}
\begin{proof}	
Let $k\ge 4$ and $t=\lfloor{\frac{k}{4}}\rfloor$.
We define the election $E=(N,C,(A_v)_{v\in N},k)$ as follows. 
\begin{align*}
	& C=C_1\cup C_2\cup D_{k-2},\quad
	C_1=\{c_{1,1},\ldots, c_{1,t+1}\}, \quad C_2=\{c_{2,1},\ldots, c_{2,k}\};\\
	& N=V_1 \cup V_2 \cup\bigcup_{j\in [k]} S_j \cup U, \text{ where }
	V_\ell=\{v_{\ell,1},\ldots, v_{\ell,t}\}, \ell\in [2], \quad
 |U|=\left\lfloor{\frac{k^2}{4}-\frac{k}{2}}\right\rfloor,\\ 
    &|S_j|=t, j\in [k].
\end{align*}
 The approval sets of voters $v_{1, i}$ and $v_{2, i}$, $i\in [t]$, are given by
\begin{align*}
&A_{v_{1,i}}=\{c_{1,i+1},\ldots, c_{1,t+1}\}\cup\{c_{2,j}\in C_2\mid j\text{ is even}\} \text{  and  }\\ 
&A_{v_{2,i}} = \{c_{1,1},\ldots, c_{1,i}\}\cup\{c_{2,j}\in C_2\mid j\text{ is odd}\},
\end{align*}
For each $j\in [k]$ the approval set of each voter $s\in S_j$ is 
$A_{s}=\{c_{2,j},\ldots, c_{2,k}\}$.
Each $u\in U$ has approval set $A_{u}=D_{k-2}$. 
Intuitively, voters in $U$ are dummy voters and candidates in $D_{k-2}$ are dummy candidates. The sequence of swaps we will exhibit affects voters in $N\setminus U$ and candidates in $C_1\cup C_2$, but no other voters or candidates.

\subparagraph{Set-up}  Consider $\frac{n}{k^2}$-ls-PAV on the above instance
with an initial committee 
$$
W_0=D_{k-2}\cup\{c_{1,1},c_{2,1}\}.
$$
We first show that $\frac{n}{k^2}\leq \frac12$. Indeed, we have
\begin{align*}
    n=|V_1|+|V_2|+\sum_{j=1}^k |S_j|+|U|=
    \left\lfloor\frac{k}{4}\right\rfloor+ 
    \left\lfloor\frac{k}{4}\right\rfloor+
    k\times \left\lfloor\frac{k}{4}\right\rfloor+
    \left\lfloor\frac{k^2}{4}-\frac{k}{2}\right\rfloor\leq \frac{k^2}{2}.
\end{align*}
In what follows, we will say that a swap $(a,b)$ is a \textit{good swap for} $W$ if $\Delta(W,a,b)\geq\frac{n}{k^2}$; we will say that $(a, b)$ is a \textit{good swap} if $W$ is clear from the context.
As we have argued that $\frac{n}{k^2}\leq \frac12$, 
every valid swap that increases the PAV score by at least $\frac{1}{2}$ 
is a good swap.

\subparagraph{Sequence of Swaps} Define 
$\textbf{\text{Y}}=\oplus_{i=1}^t(c_{1,i}, c_{1,i+1})$
and let ${\mathbf Z}_j=\mathbf Y$ if $j$ is odd and ${\mathbf Z}_j={\mathbf Y}^{-1}$
if $j$ is even. Let
$\textbf{X}=\oplus_{j=1}^{k-1}\left( {\mathbf Z}_j\oplus (c_{2,j},c_{2,j+1})\right)$, 
and note that 
$|\textbf{Y}|=|\textbf{\text{Y}}^{-1}|=t$, 
so
$|\textbf{X}|=(k-1)\cdot (t+1)=\Omega(k^2)$.
We will argue that all swaps in \textbf{X} are good.
\\

To this end, we split up the analysis into the following four claims.
\begin{enumerate}[(i)]
\item For committee $W=D_{k-2}\cup\{c_{1,i}, c_{2,j}\}$, if $j$ is odd and $i\leq t$, $(c_{1,i}, c_{1,i+1})$ is a good swap.
\item For committee $W=D_{k-2}\cup\{c_{1,i}, c_{2,j}\}$, if $j$ is even and $i>1$, $(c_{1,i}, c_{1, i-1})$ is a good swap.
\item For committee $W=D_{k-2}\cup\{c_{1,t+1}, c_{2,j}\}$, if $j<k$ is odd, $(c_{2,j},c_{2,j+1})$ is a good swap.
\item For committee $W=D_{k-2}\cup\{c_{1,1}, c_{2,j}\}$, if $j<k$ is even, $(c_{2,j},c_{2,j+1})$ is a good swap.
\end{enumerate}
Together, these four claims imply that $\textbf{X}$ is a sequence of good swaps.
Indeed, suppose that $j$ is odd, and consider 
the committee $W=D_{k-2}\cup\{c_{1,1}, c_{2,j}\}$. By Claim~(i), 
$\textbf{Z}_j = \textbf{Y}$ is a sequence of good swaps. 
After executing \textbf{Y} we obtain a committee $W=D_{k-2}\cup\{c_{1,t+1}, c_{2,j}\}$ satisfying the condition in Claim~(iii). Hence, if $j<k$, then $(c_{2,j},c_{2,j+1})$ is a good swap. This swap results in a committee $W=D_{k-2}\cup\{c_{1,t+1}, c_{2,j+1}\}$ satisfying the condition in Claim~(ii). This implies that $\textbf{Z}_{j+1} = \textbf{Y}^{-1}$ is a sequence of good swaps, resulting in a committee $W=D_{k-2}\cup\{c_{1,1}, c_{2,j+1}\}$. This committee, in turn, is as described in (iv), so if $j+1<k$, then $(c_{2,j+1},c_{2,j+2})$ a good swap. This results in $W=D_{k-2}\cup\{c_{1,1}, c_{2,j+2}\}$, which again satisfies the condition in (i). As this reasoning applies 
to all odd values of $j$, including $j=1$ (which corresponds to our starting point $W_0$), we can conclude that the sequence \textbf{X} is a sequence of good swaps.
It remains to prove Claims~(i)--(iv).

\subparagraph{Claim (i)} 
Suppose $W=D_{k-2}\cup\{c_{1,i}, c_{2,j}\}$ where $j\in [k]$ is odd and $i\in [t]$.
Then $|A_{v_{2,i}}\cap W|=2$ and $|A_{v_{1,i}}\cap W|=0$. Moreover, 
$v_{2,i}$ approves $c_{1,i}$ and not $c_{1,i+1}$, and conversely for $v_{1,i}$, while every other voter approves either both or neither of $c_{1,i}$ and $c_{1,i+1}$. 
We conclude that $(c_{1, i},c_{1,i+1})$ is a good swap, because
$\Delta(W,c_{1,i},c_{1,i+1})= +1-\frac{1}{2}=\frac{1}{2}$.

\subparagraph{Claim (ii)}
Suppose $W=D_{k-2}\cup\{c_{1,i}, c_{2,j}\}$, where $1< i \leq t+1$ and $j\in [k]$ is even. 
Then $|A_{v_{1,i-1}}\cap W|=2$, $|A_{v_{2,i-1}}\cap W|=0$, and every other voter approves either both of $c_{1,i-1}$ and $c_{1,i}$ or neither of them. Thus, $\Delta(W,c_{1,i},c_{1,i-1})= +1-\frac{1}{2}=\frac{1}{2}$.

\subparagraph{Claim (iii)}
Suppose $W=D_{k-2}\cup\{c_{1,t+1}, c_{2,j}\}$, where $j<k$ is odd. Each voter in $S_{j+1}$ approves $c_{2,j+1}$ and not $c_{2,j}$. Every voter in $V_2$ approves $c_{2,j}$, but not $c_{2,j+1}$, while every voter in $V_1$ approves $c_{2,j+1}$ and not $c_{2,j}$. By construction, the remaining voters (i.e., the voters in $S_\ell$, $\ell\neq j+1$, and the voters in $U$) approve either both of $c_{2,j}$ and $c_{2,j+1}$ or neither.
For $s\in S_{j+1}$, their satisfaction is $|A_s\cap W|=0$. For every $i\in[t]$ we have $|A_{v_{1,i}}\cap W|=1$ as $v_{1,i}$ only approves $c_{1,t+1}$ in $W$, and $|A_{v_{2,i}}\cap W|=1$, because $v_{2,i}$ only approves $c_{2,j}$ as $j$ is odd.
Thus, 
$\Delta(W,c_{2,j},c_{2,j+1})=+|S_{j+1}|+\frac{1}{2}\cdot|V_1|-|V_2|=t+\frac{t}{2}-t= \frac{t}{2}\geq \frac{1}{2}$,
provided $k\ge 4$. This shows that $(c_{2,j},c_{2,j+1})$ is a good swap. 

\subparagraph{Claim (iv)}
Suppose $W=D_{k-2}\cup\{c_{1,1}, c_{2,j}\}$, where $j<k$ is even. Similarly to the proof of Claim~(iii), we can show that 
$\Delta(W,c_{2,j},c_{2,j+1})\geq \frac{1}{2}$. 
This concludes the proof.
\end{proof}



\section{Main Result}\label{sec:mainresult}
We are now ready to present our main result.
\begin{restatable}[\appsymb]{theorem}{zerolow}\label{thm:zerolow}
    For every $k\geq 1$ there exists a committee election with $\mathrm{poly}(k)$ voters, 
a committee $W_0$, $|W_0|=k$, and a sequence of $\Omega(k^{\log k})$ swaps starting from $W_0$ such that each swap in this sequence strictly increases the PAV score.
\end{restatable}

The proof of our main lower bound is in many ways similar to the proof of Theorem~\ref{thm:threshold}. We start 
by giving a high-level overview of the proof, and then describe our construction; the proof of correctness is mostly relegated to the extended version of the paper.
We first consider the adversarial setting;
in Section~\ref{sec:nabr} we will show how to extend our proof to lexicographic better response.
 
 \subsection{Proof Overview}
 We construct a committee $W$ of size $k$ and a sequence of swaps of length $\Omega(k^{\log k})$ such that $0^+$-ls-PAV may execute this sequence when initialised on $W$. Just as in the proof of Theorem~\ref{thm:threshold}, this sequence of swaps leaves most members of the initial committee untouched. Most of the action takes place in 
 the first $k_1=O(\log k)$ committee spots; the candidates in the remaining spots stay in place throughout the entire sequence of swaps. 
For each $i=1, \dots, k_1$, the
committee spot $i$ is assigned its own set of candidates $C_i$, so that in the sequence we construct, a candidate in $C_i$ can only be replaced by another candidate from~$C_i$.

Furthermore, each of the first $k_1$ committee positions has its corresponding set of voters, and between different committee positions, initially the voters' satisfaction is very unequal: voters corresponding to the first position are most happy, and voters corresponding to later committee positions are increasingly unhappy. The sequence of swaps will again start by making the most happy voters happier and then move on to less happy voters, making them better off.\textit{ Thereby it will undo all the work it has done so far and will have to repeat it.}

Consider Figure~\ref{swapseq}, where each board represents a committee. The figure should be read left to right and top to bottom, where the next board/committee results from the previous one if certain swaps are made. More precisely, within a board each square corresponds to a candidate, so that $C_i$ is the set of candidates in the $i$-th column.
The coloured squares mark the candidates that are in the committee, and arrows between squares indicate swaps between candidates. The result of the swap(s) can be seen in the next board.


\begin{figure*}[h!]
\centering
\setlength{\tabcolsep}{4pt}
\renewcommand{\arraystretch}{1}
\begin{tabularx}{\textwidth}{@{} *{3}{c} @{}}
\begin{tabular}{c|c|c|c|c}
      i/j& $1$ & $2$ & $3$ & $4$\\ \hline
      1 & \tikzmark{1611}  &\tikzmark{1621}  &\tikzmark{1631}  & \tikzmark{1641} \cfour \\ \hline
      2 &\tikzmark{1612}  &   \tikzmark{1622}  & \tikzmark{1632} & \tikzmark{1642} \\ \hline
     3 &\cone\tikzmark{1613}&\ctwo  \tikzmark{1623}  &\tikzmark{1633}\cthree& \tikzmark{1643}  \\
\end{tabular}

\begin{tikzpicture}[overlay,remember picture, shorten >=-3pt]
\draw[->] (pic cs:1641) -- (pic cs:1642) ;
\end{tikzpicture}
\begin{tabular}{c|c|c|c|c}
      i/j& $1$ & $2$ & $3$ & $4$\\ \hline
      1 & \tikzmark{1711}  &\tikzmark{1721}  &\tikzmark{1731}  & \tikzmark{1741}  \\ \hline
      2 &\tikzmark{1712}  &   \tikzmark{1722}  & \tikzmark{1732} & \tikzmark{1742} \cfour\\ \hline
     3 &\cone\tikzmark{1713}&\ctwo  \tikzmark{1723}  &\tikzmark{1733}\cthree& \tikzmark{1743}  \\
\end{tabular}

\begin{tikzpicture}[overlay,remember picture, shorten >=-3pt]
\draw[->] (pic cs:1733) -- (pic cs:1732) ;
\end{tikzpicture}
\begin{tabular}{c|c|c|c|c}
      i/j& $1$ & $2$ & $3$ & $4$\\ \hline
      1 & \tikzmark{1811}  &\tikzmark{1821}  &\tikzmark{1831}  & \tikzmark{1841}  \\ \hline
      2 &\tikzmark{1812}  &   \tikzmark{1822}  & \tikzmark{1832} \cthree & \tikzmark{1842} \cfour\\ \hline
     3 &\cone\tikzmark{1813}&\ctwo  \tikzmark{1823}  &\tikzmark{1833}& \tikzmark{1843}  \\
\end{tabular}

\begin{tikzpicture}[overlay,remember picture, shorten >=-3pt]
\draw[->] (pic cs:1823) -- (pic cs:1822) ;
\end{tikzpicture}
\begin{tabular}{c|c|c|c|c}
      i/j& $1$ & $2$ & $3$ & $4$\\ \hline
      1 & \tikzmark{1911}  &\tikzmark{1921}  &\tikzmark{1931}  & \tikzmark{1941}  \\ \hline
      2 &\tikzmark{1912}  &   \ctwo\tikzmark{1922}  & \tikzmark{1932} \cthree & \tikzmark{1942} \cfour\\ \hline
     3 &\cone\tikzmark{1913}&  \tikzmark{1923}  &\tikzmark{1933}& \tikzmark{1943}  \\
\end{tabular}

\begin{tikzpicture}[overlay,remember picture, shorten >=-3pt]
\draw[->] (pic cs:1913) -- (pic cs:1912) ;
\draw[->] (pic cs:1912) -- (pic cs:1911) ;
\end{tikzpicture}\\\\
\begin{tabular}{c|c|c|c|c}
      i/j& $1$ & $2$ & $3$ & $4$\\ \hline
      1 &\cone \tikzmark{2011}  &\tikzmark{2021}  &\tikzmark{2031}  & \tikzmark{2041}  \\ \hline
      2 &\tikzmark{2012}  &   \ctwo\tikzmark{2022}  & \tikzmark{2032} \cthree & \tikzmark{2042} \cfour\\ \hline
     3 &\tikzmark{2013}&  \tikzmark{2023}  &\tikzmark{2033}& \tikzmark{2043}  \\
\end{tabular}

\begin{tikzpicture}[overlay,remember picture, shorten >=-3pt]
\draw[->] (pic cs:2022) -- (pic cs:2021) ;
\end{tikzpicture}

\begin{tabular}{c|c|c|c|c}
      i/j& $1$ & $2$ & $3$ & $4$\\ \hline
      1 &\cone \tikzmark{2111}  &\ctwo\tikzmark{2121}  &\tikzmark{2131}  & \tikzmark{2141}  \\ \hline
      2 &\tikzmark{2112}  &   \tikzmark{2122}  & \tikzmark{2132} \cthree & \tikzmark{2142} \cfour\\ \hline
     3 &\tikzmark{2113}&  \tikzmark{2123}  &\tikzmark{2133}& \tikzmark{2143}  \\
\end{tabular}

\begin{tikzpicture}[overlay,remember picture, shorten >=-3pt]
\draw[->] (pic cs:2111) -- (pic cs:2112) ;
\draw[->] (pic cs:2112) -- (pic cs:2113) ;
\end{tikzpicture}

\begin{tabular}{c|c|c|c|c}
      i/j& $1$ & $2$ & $3$ & $4$\\ \hline
      1 & \tikzmark{2211}  &\ctwo\tikzmark{2221}  &\tikzmark{2231}  & \tikzmark{2241}  \\ \hline
      2 &\tikzmark{2212}  &   \tikzmark{2222}  & \tikzmark{2232} \cthree & \tikzmark{2242} \cfour\\ \hline
     3 &\cone\tikzmark{2213}&  \tikzmark{2223}  &\tikzmark{2233}& \tikzmark{2243}  \\
\end{tabular}

\begin{tikzpicture}[overlay,remember picture, shorten >=-3pt]
\draw[->] (pic cs:2232) -- (pic cs:2231) ;
\end{tikzpicture}

\begin{tabular}{c|c|c|c|c}
      i/j& $1$ & $2$ & $3$ & $4$\\ \hline
      1 & \tikzmark{2311}  &\ctwo\tikzmark{2321}  &\tikzmark{2331}  \cthree  & \tikzmark{2341}  \\ \hline
      2 &\tikzmark{2312}  &   \tikzmark{2322}  & \tikzmark{2332}& \tikzmark{2342} \cfour\\ \hline
     3 &\cone\tikzmark{2313}&  \tikzmark{2323}  &\tikzmark{2333}& \tikzmark{2343}  \\
\end{tabular}

\begin{tikzpicture}[overlay,remember picture, shorten >=-3pt]
\draw[->] (pic cs:2321) -- (pic cs:2322)  ;
\end{tikzpicture}

\begin{tabular}{c|c|c|c|c}
      i/j& $1$ & $2$ & $3$ & $4$\\ \hline
      1 & \tikzmark{2411}  &\tikzmark{2421}  &\tikzmark{2431}  \cthree  & \tikzmark{2441}  \\ \hline
      2 &\tikzmark{2412}  &\ctwo   \tikzmark{2422}  & \tikzmark{2432}& \tikzmark{2442} \cfour\\ \hline
     3 &\cone\tikzmark{2413}&  \tikzmark{2423}  &\tikzmark{2433}& \tikzmark{2443}  \\
\end{tabular}

\begin{tikzpicture}[overlay,remember picture, shorten >=-3pt]
\draw[->] (pic cs:2413) -- (pic cs:2412) ;
\draw[->] (pic cs:2412) -- (pic cs:2411) ;
\end{tikzpicture}\\\\

\begin{tabular}{c|c|c|c|c}
      i/j& $1$ & $2$ & $3$ & $4$\\ \hline
      1 & \cone\tikzmark{2511}  &\tikzmark{2521}  &\tikzmark{2531}  \cthree  & \tikzmark{2541}  \\ \hline
      2 &\tikzmark{2512}  &\ctwo   \tikzmark{2522}  & \tikzmark{2532}& \tikzmark{2542} \cfour\\ \hline
     3 &\tikzmark{2513}&  \tikzmark{2523}  &\tikzmark{2533}& \tikzmark{2543}  \\
\end{tabular}

\begin{tikzpicture}[overlay,remember picture, shorten >=-3pt]
\draw[->] (pic cs:2522) -- (pic cs:2523) ;
\end{tikzpicture}

\begin{tabular}{c|c|c|c|c}
      i/j& $1$ & $2$ & $3$ & $4$\\ \hline
      1 & \cone\tikzmark{2611}  &\tikzmark{2621}  &\tikzmark{2631}  \cthree  & \tikzmark{2641}  \\ \hline
      2 &\tikzmark{2612}  &   \tikzmark{2622}  & \tikzmark{2632}& \tikzmark{2642} \cfour\\ \hline
     3 &\tikzmark{2613}&\ctwo  \tikzmark{2623}  &\tikzmark{2633}& \tikzmark{2643}  \\
\end{tabular}

\begin{tikzpicture}[overlay,remember picture, shorten >=-3pt]
\draw[->] (pic cs:2611) -- (pic cs:2612) ;
\draw[->] (pic cs:2612) -- (pic cs:2613) ;
\end{tikzpicture}

\begin{tabular}{c|c|c|c|c}
      i/j& $1$ & $2$ & $3$ & $4$\\ \hline
      1 &\tikzmark{2711}  &\tikzmark{2721}  &\tikzmark{2731}  \cthree  & \tikzmark{2741}  \\ \hline
      2 &\tikzmark{2712}  &   \tikzmark{2722}  & \tikzmark{2732}& \tikzmark{2742} \cfour\\ \hline
     3 & \cone\tikzmark{2713}&\ctwo  \tikzmark{2723}  &\tikzmark{2733}& \tikzmark{2743}  \\
\end{tabular}

\begin{tikzpicture}[overlay,remember picture, shorten >=-3pt]
\draw[->] (pic cs:2742) -- (pic cs:2743) ;
\end{tikzpicture}

\begin{tabular}{c|c|c|c|c}
      i/j& $1$ & $2$ & $3$ & $4$\\ \hline
      1 &\tikzmark{2811}  &\tikzmark{2821}  &\tikzmark{2831}  \cthree  & \tikzmark{2841}  \\ \hline
      2 &\tikzmark{2812}  &   \tikzmark{2822}  & \tikzmark{2832}& \tikzmark{2842} \\ \hline
     3 & \cone\tikzmark{2813}&\ctwo  \tikzmark{2823}  &\tikzmark{2833}&\cfour \tikzmark{2843}  \\
\end{tabular}

\begin{tikzpicture}[overlay,remember picture, shorten >=-3pt]
\draw[->] (pic cs:2831) -- (pic cs:2832) ;
\end{tikzpicture}

\begin{tabular}{c|c|c|c|c}
      i/j& $1$ & $2$ & $3$ & $4$\\ \hline
      1 &\tikzmark{2911}  &\tikzmark{2921}  &\tikzmark{2931}   & \tikzmark{2941}  \\ \hline
      2 &\tikzmark{2912}  &   \tikzmark{2922}  & \tikzmark{2932}\cthree& \tikzmark{2942} \\ \hline
     3 & \cone\tikzmark{2913}&\ctwo  \tikzmark{2923}  &\tikzmark{2933}&\cfour \tikzmark{2943}  \\
\end{tabular}

\begin{tikzpicture}[overlay,remember picture, shorten >=-3pt]
\draw[->] (pic cs:2923) -- (pic cs:2922) ;
\end{tikzpicture}\\\\

\begin{tabular}{c|c|c|c|c}
      i/j& $1$ & $2$ & $3$ & $4$\\ \hline
      1 &\tikzmark{3011}  &\tikzmark{3021}  &\tikzmark{3031}   & \tikzmark{3041}  \\ \hline
      2 &\tikzmark{3012}  & \ctwo   \tikzmark{3022}  & \tikzmark{3032}\cthree& \tikzmark{3042} \\ \hline
     3 & \cone\tikzmark{3013}& \tikzmark{3023}  &\tikzmark{3033}&\cfour \tikzmark{3043}  \\
\end{tabular}

\begin{tikzpicture}[overlay,remember picture, shorten >=-3pt]
\draw[->] (pic cs:3013) -- (pic cs:3012) ;
\draw[->] (pic cs:3012) -- (pic cs:3011) ;
\end{tikzpicture}

\begin{tabular}{c|c|c|c|c}
      i/j& $1$ & $2$ & $3$ & $4$\\ \hline
      1 &\cone \tikzmark{3111}  &\tikzmark{3121}  &\tikzmark{3131}   & \tikzmark{3141}  \\ \hline
      2 &\tikzmark{3112}  & \ctwo   \tikzmark{3122}  & \tikzmark{3132}\cthree& \tikzmark{3142} \\ \hline
     3 & \tikzmark{3113}& \tikzmark{3123}  &\tikzmark{3133}&\cfour \tikzmark{3143}  \\
\end{tabular}

\begin{tikzpicture}[overlay,remember picture, shorten >=-3pt]
\draw[->] (pic cs:3122) -- (pic cs:3121) ;
\end{tikzpicture}

\begin{tabular}{c|c|c|c|c}
      i/j& $1$ & $2$ & $3$ & $4$\\ \hline
      1 &\cone \tikzmark{3211}  & \ctwo \tikzmark{3221}  &\tikzmark{3231}   & \tikzmark{3241}  \\ \hline
      2 &\tikzmark{3212}  &  \tikzmark{3222}  & \tikzmark{3232}\cthree& \tikzmark{3242} \\ \hline
     3 & \tikzmark{3213}& \tikzmark{3223}  &\tikzmark{3233}&\cfour \tikzmark{3243}  \\
\end{tabular}

\begin{tikzpicture}[overlay,remember picture, shorten >=-3pt]
\draw[->] (pic cs:3211) -- (pic cs:3212) ;
\draw[->] (pic cs:3212) -- (pic cs:3213) ;
\end{tikzpicture}

\begin{tabular}{c|c|c|c|c}
      i/j& $1$ & $2$ & $3$ & $4$\\ \hline
      1 & \tikzmark{3311}  & \ctwo \tikzmark{3321}  &\tikzmark{3331}   & \tikzmark{3341}  \\ \hline
      2 &\tikzmark{3312}  &  \tikzmark{3322}  & \tikzmark{3332}\cthree& \tikzmark{3342} \\ \hline
     3 &\cone \tikzmark{3313}& \tikzmark{3323}  &\tikzmark{3333}&\cfour \tikzmark{3343}  \\
\end{tabular}

\begin{tikzpicture}[overlay,remember picture, shorten >=-3pt]
\draw[->] (pic cs:3332) -- (pic cs:3333) ;
\end{tikzpicture}

\begin{tabular}{c|c|c|c|c}
      i/j& $1$ & $2$ & $3$ & $4$\\ \hline
      1 & \tikzmark{3411}  & \ctwo \tikzmark{3421}  &\tikzmark{3431}   & \tikzmark{3441}  \\ \hline
      2 &\tikzmark{3412}  &  \tikzmark{3422}  & \tikzmark{3432}& \tikzmark{3442} \\ \hline
     3 &\cone \tikzmark{3413}& \tikzmark{3423}  &\tikzmark{3433} \cthree&\cfour \tikzmark{3443}  \\
\end{tabular}

\begin{tikzpicture}[overlay,remember picture, shorten >=-3pt]
\draw[->] (pic cs:3421) -- (pic cs:3422) ;
\end{tikzpicture}\\\\
\begin{tabular}{c|c|c|c|c}
      i/j& $1$ & $2$ & $3$ & $4$\\ \hline
      1 & \tikzmark{3511}  & \tikzmark{3521}  &\tikzmark{3531}   & \tikzmark{3541}  \\ \hline
      2 &\tikzmark{3512}  & \ctwo  \tikzmark{3522}  & \tikzmark{3532}& \tikzmark{3542} \\ \hline
     3 &\cone \tikzmark{3513}& \tikzmark{3523}  &\tikzmark{3533} \cthree&\cfour \tikzmark{3543}  \\
\end{tabular}

\begin{tikzpicture}[overlay,remember picture, shorten >=-3pt]
\draw[->] (pic cs:3513) -- (pic cs:3512) ;
\draw[->] (pic cs:3512) -- (pic cs:3511) ;
\end{tikzpicture}

\begin{tabular}{c|c|c|c|c}
      i/j& $1$ & $2$ & $3$ & $4$\\ \hline
      1 &\cone  \tikzmark{3611}  & \tikzmark{3621}  &\tikzmark{3631}   & \tikzmark{3641}  \\ \hline
      2 &\tikzmark{3612}  & \ctwo  \tikzmark{3622}  & \tikzmark{3632}& \tikzmark{3642} \\ \hline
     3 &\tikzmark{3613}& \tikzmark{3623}  &\tikzmark{3633} \cthree&\cfour \tikzmark{3643}  \\
\end{tabular}

\begin{tikzpicture}[overlay,remember picture, shorten >=-3pt]
\draw[->] (pic cs:3622) -- (pic cs:3623) ;
\end{tikzpicture}

\begin{tabular}{c|c|c|c|c}
      i/j& $1$ & $2$ & $3$ & $4$\\ \hline
      1 &\cone  \tikzmark{3711}  & \tikzmark{3721}  &\tikzmark{3731}   & \tikzmark{3741}  \\ \hline
      2 &\tikzmark{3712}  &  \tikzmark{3722}  & \tikzmark{3732}& \tikzmark{3742} \\ \hline
     3 &\tikzmark{3713}&  \ctwo\tikzmark{3723}  &\tikzmark{3733} \cthree&\cfour \tikzmark{3743}  \\
\end{tabular}

\begin{tikzpicture}[overlay,remember picture, shorten >=-3pt]
\draw[->] (pic cs:3711) -- (pic cs:3712) ;
\draw[->] (pic cs:3712) -- (pic cs:3713) ;

\end{tikzpicture}

\begin{tabular}{c|c|c|c|c}
      i/j& $1$ & $2$ & $3$ & $4$\\ \hline
      1 &  \tikzmark{3811}  & \tikzmark{3821}  &\tikzmark{3831}   & \tikzmark{3841}  \\ \hline
      2 &\tikzmark{3812}  &  \tikzmark{3822}  & \tikzmark{3832}& \tikzmark{3842} \\ \hline
     3 &\cone\tikzmark{3813}&  \ctwo\tikzmark{3823}  &\tikzmark{3833} \cthree&\cfour \tikzmark{3843}  \\
\end{tabular}
\end{tabularx}
\caption{We illustrate the sequence of swaps in Theorem~\ref{thm:zerolow} by a small example. Each one of the $23$ $3\times 4$ grids shows 12 candidates, one for each empty square, the 4 coloured ones indicating candidates currently in the committee. The empty squares in column $i$ are the candidates in $C_i$; they are ordered as $c_{i,1},c_{i,2},c_{i,3}$ from top to bottom. 
	We omit the dummy candidates from this picture, and let $k_1=4$ and $t+1=3$, as larger $t$ is only necessary for the sequence length in the proof. An arrow indicates the swap that will transform the current committee into the next committee. The top left initial committee is $\{c_{1,3},c_{2,3},c_{3,3},c_{4,1}\}$ and the bottom right final committee is $\{c_{1,3},c_{2,3},c_{3,3},c_{4,3}\}$. }\label{swapseq}
\end{figure*}
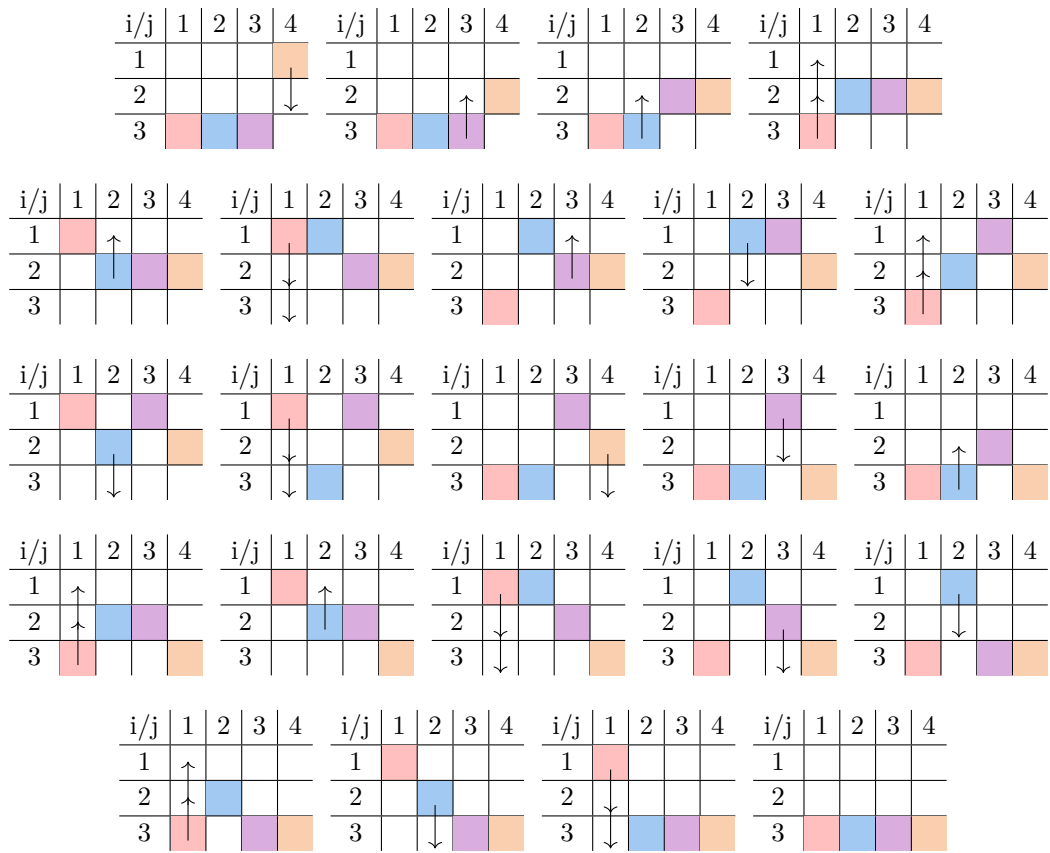

The pattern in Figure~\ref{swapseq} may not be immediately evident; the reader may want to revisit it after having read the proof of Theorem~\ref{thm:zerolow}. For now, we use Figure~\ref{swapseq} to illustrate how we build up our instance in several steps, creating increasingly larger building blocks.
 Observe that in Figure~\ref{swapseq} swaps only occur along a column (corresponding to candidates $C_i$ for some $i$): indeed, as we mentioned earlier, in our construction  swaps can only replace candidates in $C_i$ with other candidates in $C_i$. 
 \begin{enumerate}
 \item 
 Zooming in on a single swap, the voters responsible for this swap form an atomic building block of our construction. This building block, Election $E(j,k)$, is given in Section~\ref{sec:level1}. Here, $j$ is carefully picked to depend on the column $i$ (i.e., the $i$-th spot on the committee). In Lemma~\ref{atomgroup} we show that the corresponding swap increases the PAV score by exactly 
 $$
 \delta(j,k)=\frac{j!}{\prod_{j'=0}^{j}(k-j')}.
 $$
\item 
Zooming out to just the $i$-th column, the election responsible for the dynamics along the $i$-th column is $E^t(j,k)$, where $t=|C_i|-1$. 
We discuss how to construct $E^t(j, k)$ from $E(j,k)$ in Section~\ref{sec:level2}.
\item 
Finally, the entire board roughly corresponds to election $E$, constructed out of the building blocks $E^t(j,k)$ (Section~\ref{sec:level3}). 
\end{enumerate}
With the constructed election $E$ in hand, in Section~\ref{sec:abr} we exhibit an initial committee and a sequence of good swaps of length superpolynomial in $k$, such that, starting from this committee, $0^+$-ls-PAV under adversarial better response executes this sequence of swaps.
Finally, in Section~\ref{sec:nabr} we show how to modify this instance to establish that, even when the improving swaps are selected according to a fixed pivoting rule (rather than adversarially), $0^+$-ls-PAV may still make super-polynomially many swaps. To this end, we prove that a long subsequence of the swap sequence from Theorem~\ref{thm:zerolow} is preserved by the pivoting rule.


 \subsection{First Steps: Election $\boldsymbol{E(j,k)}$}\label{sec:level1}

We now introduce a family of elections that form the smallest building blocks of our instance.
We will frequently say that an election $E$ has structure $X$ if it is isomorphic to election $X$.
As in the proof of Theorem~\ref{thm:threshold}, we write $D_\ell=\{d_1, \dots ,d_\ell\}$
to denote a set of $\ell$ dummy candidates.
For committee size $k\in \mathbb{N}$, we will use induction on $j$ to construct elections $F(j,k)=(N,C,(A_v)_{v\in N},k)$ with $|N|=2^j$
and $C=D_{k-1}\cup \{a,b\}$ for each $1\le j<k$. 

\subparagraph{Construction 1} [Election $F(j,k)$]\label{fjk} 
Fix $j, k\in \mathbb{N}$ with $1\le j< k$. For $j=1$, let $F(j,k)=(\{1, 2\},C,(A_v)_{v=1, 2},k)$,
where $C=D_{k-1}\cup\{a,b\}$, $A_1=D_{k-1}\cup\{a\}$ and $A_2=D_{k-2}\cup\{b\}$.

For $j>1$ (and $k\geq j+1$), we construct $F(j,k)$ as follows. Consider elections $F(j-1,k)=(N_1,C_1,(A_v)_{v\in N_1},k)$ and $F(j-1,k-1)=(N_2,C_2,(A_v)_{v\in N_2},k-1)$,
where $N_1\cap N_2= \varnothing$ (we relabel the voters to ensure they are distinct), and $C_1=D_{k-1}\cup \{a_1,b_1\}$, $C_2=D_{k-2}\cup \{a_2,b_2\}$.
We set $C=D_{k-1}\cup\{a,b\}$, $N=N_1\sqcup N_2$, 
and modify the voters' ballots as follows: for each $v\in N_1$ we replace each occurrence
of $a_1$ and $b_1$ in $A_v$ with $b$ and $a$, respectively, and 
for each $v\in N_2$ we replace each occurrence
of $a_2$ and $b_2$ in $A_v$ with $a$ and $b$, respectively.
We then define $F(j,k)=(N, C,(A_v)_{v\in N},k)$.

Observe that for each $k\ge 1$ the number of voters in $F(j,k)$ is exactly $2^j$; this follows easily by induction since there are two voters in election $F(1,k)$, and for $j>1$, elections $F(j-1,k)$ and $F(j-1,k-1)$ have disjoint sets of voters of size $2^{j-1}$ each.

\subparagraph{Construction 2} [Election $E(j,k)$] Election $E(j,k)$ is built similarly to $F(j,k)$. We start with two copies of $F(j-1, k-1)$, merge them as in the construction for $F(j, k)$, 
introduce two new candidates $x$ and $y$, and make all voters from the first copy approve $x$
and all voters from the second copy approve~$y$. 

Formally, consider two disjoint copies of $F(j-1, k-1)$ given by
$(N_1,C_1,(A_v)_{v\in N_1},k)$ 
and $(N_2,C_2,(A_v)_{v\in N_2},k)$, where 
$C_1=D_{k-2}\cup \{a_1,b_1\}$, $C_2=D_{k-2}\cup \{a_2,b_2\}$,
and define a new election $(N,C,(A'_v)_{v\in N_1\cup N_2},k)$ as follows. 
Set $N=N_1\sqcup N_2$ and $C=D_{k-2}\cup\{x,y\}\cup\{a,b\}$. 
Again, modify the voters' ballots accordingly:
for each $v\in N_1$ replace each occurrence
of $a_1$ and $b_1$ in $A_v$ with $b$ and $a$, respectively, and 
for each $v\in N_2$ replace each occurrence
of $a_2$ and $b_2$ in $A_v$ with $a$ and $b$, respectively.
Then, set $A'_v=A_v\cup \{x\}$ if $v\in N_1$ and $A'_v=A_v\cup \{y\}$ if $v\in N_2$.
Let $s:N_1\mapsto N_2$ be the natural bijection between the sets of voters 
in the two isomorphic elections that $E(j,k)$ is built from; 
we will say that $s(v)$ is the counterpart of $v$, and $v$ is the counterpart of $s(v)$.
We collect a few simple properties of $E(j, k)$ in the following observation.
 \begin{proposition}\label{mapping}
 For all $1\le j< k$, the election $E(j,k)$ has the following properties:
    \begin{enumerate}
    \item $A_v\cap D=A_{s(v)}\cap D$,     
    \item $v$ approves $a$ (resp., $b$) if and only if $s(v)$ approves $b$ 
    (resp., $a$), and
    \item $v$ approves $x$ (resp., $y$) if and only if $s(v)$ approves $y$ 
    (resp., $x$).
    \end{enumerate}
\end{proposition}

Consider the committee $W=D_{k-2} \cup\{x, a\}$. The election $E(j,k)$ satisfies two important properties with respect to $W$, 
which will be needed in the proof of 
Theorem~\ref{thm:zerolow}. We state them in the following lemma. To make the lemma easier to use, we define $\delta:\mathbb{N}\times \mathbb{N}\mapsto \mathbb{R}$ as 
$$
\delta(j,k)=
\frac{j!}{\prod_{j'=0}^{j}(k-j')}. 
$$
\begin{restatable}[\appsymb]{lemma}{atomgroup}\label{atomgroup}    
For all $1\le j< k$, the election $E(j,k)$ and the committee $W=D_{k-2} \cup\{x, a\}$ 
have the following properties:
\begin{enumerate} 
    \item $\Delta(W,a,b)=\delta(j,k), $ and
    \item for every voter $v\in N$ we have $|A_v\cap W|\geq k-(j+1)$.
\end{enumerate}\end{restatable}

Proposition~\ref{mapping}
together with Lemma~\ref{atomgroup} imply the following corollary.

\begin{corollary}\label{goodswap}
    Consider election $E(j, k)$ with $1\le j< k$
    and committees $W=D_{k-2}\cup\{x,a\}$ and $W'=D_{k-2}\cup \{y,b\}$.
    It holds that $\Delta(W,a,b)=\delta(j,k)$, and hence $(a,b)$ is a good swap.
    Moreover, $\Delta(W',b,a)=\delta(j,k)$, and hence $(b,a)$ is a good swap.   
\end{corollary}


\subsection{Level up: Election $\boldsymbol{E^t(j,k)}$}\label{sec:level2}
We now combine $t$ copies of the election $E(j,k)$ into a single election, which we will call $E^t(j,k)$. This election is the building block in our final construction that is responsible for the up-and-down movement within columns, as shown in \Cref{swapseq}.

\subparagraph{Construction 3} [Election $E^t(j,k)$].
Let $t \in \mathbb{N}$. For each $i\in [t]$, we consider an election $E_i=(N_i,C_i,(A_v)_{v\in N_i},k)$, where $E_i$ has structure $E(j,k)$
with $C_i=D_{k-2}\cup\{a_i,b_i\}\cup \{x_i,y_i\}$.
For each $i= 1, \dots, t-1$, we identify $b_{i}$ with $a_{i+1}$,
and relabel $a_1,\ldots, a_t,b_t$ as $c_1,c_2,\ldots, c_{t+1}$.
Furthermore, we identify $x_i$ with $y_{i+1}$ as well as $y_i$ with $x_{i+1}$ for $i=1,\ldots t-1$, and we write $x=x_1=y_2=x_3\ldots$ and $y=y_1=x_2=y_3\ldots$.
We then set $C=\{c_1,c_2, \ldots, c_{t+1},x,y\}$.
Let $C_i^-=\{c_1, \dots, c_i\}$ and $C_i^+=\{c_{i+1}, \dots, c_{t+1}\}$.
For each $v\in N_i$, if $A_v\cap C_i^-\neq\varnothing$, 
we add all candidates in $C_i^-$ to $v$'s ballot, and if $A_v\cap C_i^+\neq\varnothing$, 
we add all candidates in $C_i^+$ to $v$'s ballot; that is, we set
$$
A'_v = A_v \cup \bigcup_{\substack{X\in\{C_i^-, C_i^+\}:\\ X\cap A_v\neq\varnothing}} X.
$$
Then, each voter $v\in N_i$ views the candidates $c_1,\ldots c_i$ as \textit{clones}: she either approves all or none of them.
Similarly, she views $c_{i+1},\ldots, c_{t+1}$ as clones, too. Finally, let $N=\sqcup_{i=1}^t N_i$,
and define $E^t(j,k)=(N,D_{k-2}\cup C, (A'_v)_{v\in N},k)$. Since each $E(j,k)$ is an election with $2^j$ voters, we can easily calculate the number of voters in $E^t(j,k)$.
\begin{proposition}\label{votenum}
The number of voters in $E^t(j,k)$ is $t2^j$.
\end{proposition}

Consider an election $E$ with structure $E^t(j,k)$ and voters $\sqcup_{i=1}^t N_i$, as in the above construction.
Since in election $E_i$ with structure $E(j,k)$ each voter $v$ in $N_i$ approves exactly one candidate from $\{a_i, b_i\}$, in $E^t(j,k)$ voter
$v$ also approves exactly one of $c_i$ and $c_{i+1}$. So, by the above construction, in $E^t(j,k)$ for each $v\in N_i$ we have either $A_v\cap C=\{c_1,\ldots, c_{i}\}$ or $A_v\cap C = \{c_{i+1},\ldots, c_{t+1}\}$. Consequently, in election $E^t(j,k)$ each of the swaps $(c_i,c_{i+1})$ and $(c_{i+1},c_{i})$ can only change the satisfaction of voters in $N_i\subset N$.

We can therefore make some useful observations regarding good swaps in election $E^t(j,k)$ with respect to the committee $W=\{d_1,\ldots,d_{k-2}, z, c\}$, where $z\in \{x,y\}$ and $c\in\{c_1, c_{t+1}\}$.

\begin{restatable}[\appsymb]{proposition}{goodseq}\label{goodseq}
\begin{enumerate}[(1)]
\item If $W_1=D_{k-2}\cup\{x,c_1\}$, then $(c_1,c_2),(c_2,c_3),\ldots, (c_t,c_{t+1})$ is a sequence of $t$ good swaps starting from $W_1$, increasing the PAV score by 
$
t \delta(j,k).$
\item If $W_2=D_{k-2}\cup \{y,c_{t+1}\}$, then $(c_{t+1},c_t),(c_t,c_{t-1}),\ldots, (c_2,c_1)$ is a sequence of $t$ good swaps
starting from $W_2$, increasing the PAV score by $
t \delta(j,k)$.
\item If $W_3=D_{k-2}\cup \{x,c_{t+1}\}$ then $\Delta(W_3,x,y)=-t\delta(j,k)$.
\item If $W_4=D_{k-2}\cup \{y,c_1\}$, then $\Delta(W_4,y,x)=-t\delta(j,k)$.
\end{enumerate}
\end{restatable}

\subsection{Final Election Instance $\boldsymbol{E}$}\label{sec:level3}
Let $k$ be the desired committee size, and
let $k_1=\lceil{\log k}\rceil$, $k_2=k-k_1$.
 We will construct an election $E=(N,C,(A_v)_{v\in N},k)$ together with an initial committee $W_0$ so that there exists a sequence of swaps of super-polynomial length starting from $W_0$.
 Briefly, $E$ is obtained by combining $k_1$ elections of the form $E^t(2j,k_2+1)$ for $j=1,\ldots, k_1$, with some modifications of the ballots.
 
\subparagraph{Constructing the Instance} 
Let $t=2\cdot\lceil\frac{k}{2}\rceil$, so that
$t\in\{k, k+1\}$ and $t$ is even. For each $i\in [k_1]$, let $C_i=\{c_{i,1},\ldots,c_{i,t+1},x_i,y_i\}$ and
\begin{align}
    E_i=(N_i,C_i\cup D_{k_2-1},(A_v)_{v\in N_i},k_2+1)\text{ with structure }E^t(2k_1-2(i-1),k_2+1).
\end{align}
The committee size of $k_2+1$ in this construction is chosen so that $k_2-1$ spots are reserved for dummy candidates $D_{k_2}$, one spot is reserved for one of $x_i$ and $y_i$, and the last spot is reserved for one of the candidates $c_{i,j}, 1\leq j \leq t+1$.
We define $E$ by merging these elections in the natural way, but we additionally modify some approvals. Informally, for all $i\in [k_1]$, we remove all candidates $x_i,y_i$ and let the candidates in $C_{i+1}\setminus\{x_{i+1},y_{i+1}\}$ take on the roles of $x_i$ and $y_i$ for voters $N_i$. That is, in election $E$ 
we modify the ballots so that 
for each $i<k_1$ it holds that 
all voters in $N_{i}$
who approve $x_i$ in $E_i$ instead approve of all $c_{i+1,j}$ with $j$ odd, and all voters in $N_{i}$ 
who approve $y_i$ in $E_i$ instead approve of all $c_{i+1,j}$ with $j$ even.  Furthermore, we add an additional dummy voter $d_{k_2}$ to the election and identify $d_{k_2}$ with $x_{k_1}$, so that every voter in $N_{k_1}$ who previously approved $x_{k_1}$ now approves $d_{k_2}$ instead.
More formally, let $E=(N,C,(A'_v)_{v\in N},k_1+k_2)$, where
\begin{align*}
    &N = \bigcup_{i=1}^{k_1}N_i, \qquad 
    C=D_{k_2}\cup\bigcup_{i=1}^{k_1}(C_i\setminus\{y_i,x_i\}), \text{ and }
\end{align*}
$$
A'_v = \begin{cases}
  A_v\cup \{c_{i+1,j} \mid  j \text{ is odd, } x_i\in A_v\text{ or }  j \text{ is even, } y_i\in A_v\}\setminus \{y_i,x_i\} & \text{ if }v\in N_{i}, i<k_1 \\
  A_v \cup \{d_{k_2}\mid x_{k_1}\in A_v\}\setminus\{x_{k_1},y_{k_1}\} & \text{ if }v\in N_{k_1}
  \end{cases}
  $$

To see why the size of $E$ is polynomial in $k$, note that
$|C|=k_2+k_1\cdot (t+1)\leq 2k \log k$, and 
by Proposition~\ref{votenum} we have $|N_i|\leq t 2^{2(\lceil{\log k}\rceil)-2(i-1)}\leq t \cdot 2^{2\lceil{\log k}\rceil}\leq 4k^2(k+1)$ and so $|N|=\sum_{i=1}^{k_1}|N_i|=\mathrm{poly}(k)$.

\subsection{Adversarial Better Response}\label{sec:abr}
We are now ready to
prove our lower bound of $\Omega(k^{\log k})$. 
We refer to our sequence of swaps as adversarial better response, because these are the swaps that an agent that points out improvements of the existing state, but acts adversarially, might choose to show us.
Let the initial committee be 
$$
W_0=(c_{1,t+1},c_{2,t+1},\ldots,c_{k_1-1,t+1}, c_{k_1,1},d_1,\ldots d_{k_2}), 
$$
so $|W_0|=k_1+k_2=k$. 
We will exhibit a sequence of good swaps that results in the final committee 
$$
(c_{1,t+1},c_{2,t+1},\ldots, c_{k_1,t+1},d_1,\ldots, d_{k_2}).
$$ 

Our basic building block is the sequence $\textbf{Y}_1=\oplus_{j=1}^{t}(c_{1,j},c_{1,j+1})$.
Let $\textbf{X}^1_1=\textbf{Y}_1$ and $\textbf{X}^0_1=\textbf{Y}^{-1}_1$.
Further, for $i>1$ define
$$
\textbf{X}^0_i=\oplus_{j=1}^{t}\left((c_{i,t-j+2},c_{i,t-j+1})\oplus\textbf{X}_{i-1}^{j-1\!\!\!\mod 2}\right), \quad
\textbf{X}^1_i=\oplus_{j=1}^{t}\left((c_{i,j},c_{i,j+1})\oplus\textbf{X}_{i-1}^{j-1\!\!\!\mod 2}\right). 
$$
Our proof shows that $0^+$-ls-PAV will perform the sequence of swaps $\textbf{X}^1_{k_1}$ when run on election $E$ (constructed in Section~\ref{sec:level3}) with initial committee $W_0$. That is, $\textbf{X}^1_{k_1}$ is a sequence of good swaps.
Further, since $|\textbf{X}^0_i|=|\textbf{X}^1_i|$, the length of $\textbf{X}^\ell_i$, $\ell\in\{0,1\}$, is $t(|\textbf{X}^\ell_{i-1}|+1)$ and $|\textbf{X}^\ell_1|=t\geq k$. Hence, $\textbf{X}^1_{k_1}$ has length $\Omega(t^{k_1})=\Omega(k^{\log k})$, i.e., it is a sequence of good swaps with super-polynomial length.

\subsection{Extension to a Fixed Pivoting Rule}\label{sec:nabr}
We adapt the proof of Theorem~\ref{thm:zerolow} to a natural non-adversarial setting. 
An intuitive method to select swaps is to consider a fixed ordering on the candidates $C=\{c_1,\ldots,c_m\}$, for example $c_1<\ldots<c_m$.
To find a good swap $(c',c)\in W\times (C\setminus W)$, we go over the candidates in $C\setminus W$, in the order suggested by $<$;
for each $c\in C\setminus W$, we go over candidates in $W$, in the order suggested by $<$, to find $c'$ such that $(c',c)$ is a good swap. That is, we consider a lexicographic ordering on pairs $(c,c')$ (where $c$ is to be added and $c'$ is to be removed from the committee) induced by the order $<$.
In light of this, we call the corresponding pivoting rule \textit{lexicographic better response}.
\begin{restatable}[\appsymb]{theorem}{lexresponse}\label{thm:nabr}

For any $k\geq 1$ there exists a committee election with $\mathrm{poly}(k)$ voters and
a committee $W_0$ such that
executing $0^+$-ls-PAV with lexicographic better response from $W_0$ results in $\Omega(k^{\log k})$ swaps.
\end{restatable}

Due to space constraints, we relegate the proof of Theorem~\ref{thm:nabr} to the extended version of the paper.
 
\section{Discussion}
We have shown that if $\varepsilon$ can be arbitrarily small, the running
time of $\varepsilon$-ls-PAV with lexicographic better response may be super-polynomial, resolving the open question of Aziz et al.~\cite{AE+18}. 
Thus, while using very small values of $\varepsilon$ would be attractive
both in terms of obtaining a more decisive rule and in terms of providing
fairness guarantees to small minorities of voters, this would come at a cost
of superpolynomial execution time in the worst case.
While a similar result for best response remains elusive, our simulations (see the extended version of the paper) shows that, at least empirically, better response is preferable to best response on both synthetic and real-world datasets. We note that our lower bound does not preclude the possibility
that an outcome of $0^+$-ls-PAV can be found in polynomial time by other means;
it is an interesting open question whether this is indeed possible. 


\bibliography{lipics-v2021-sample-article}
\appendix
\newpage
\section{Omitted Proofs}

We first establish a simple fact that will be useful for our analysis. First, note that $F(j,k-1)$ and $F(j,k)$ have the same set of voters $N=N_1\sqcup N_2$ where $|N|=2^{j}$. 

\begin{proposition}\label{prop:k-to-k-1}
    Consider $1\le j<k-1$ and elections $F(j, k-1)$ and $F(j, k)$ with voter set $N$. For each $v\in N$, denote the ballots of $v$ in $F(j, k-1)$ and $F(j, k)$ by $A_v$ and $A'_v$, respectively. Then $|A'_v\setminus A_v|=1$, and the unique candidate in $A'_v\setminus A_v$ is a dummy candidate. 
\end{proposition}
\begin{proof}

    The proof is by induction on $j$; we prove our claim for all $k>j+1$ simultaneously.
    
    For $j=1$, we have $N=\{1, 2\}$. 
    In $F(1, k)$, the voters' ballots are $A_1=D_{k-1}\cup\{a\}$ and $A_2=D_{k-2}\cup\{b\}$, whereas 
    in $F(1, k-1)$, the voters' ballots are $A'_1=D_{k-2}\cup\{a\}$ and $A'_2=D_{k-3}\cup\{b\}$ (with $D_0=\varnothing$), so our claim holds.
    
    Now, suppose our claim has been established for all $j'<j$ and all $k>j+1$.
    
    If $v\in N_1$ then in $F(j,k)$ $v$ inherits her approval ballot from $F(j-1,k)=(N_1, D_{k-1}\cup \{a,b\},A_{v\in N_1}, k)$
    and in $F(j,k-1)$ inherits her approval ballot from $F(j-1,k-1)=(N_2, D_{k-1}\cup \{a,b\},A_{v\in N_2}, k-1)$, and so applying the inductive hypothesis with $j'=j-1$ and we conclude the claim for $v\in N_1$,
    Similarly, if $v\in N_2$ then in $F(j,k)$ $v$ inherits her approval ballot from $F(j-1,k-1)=(N_1, D_{k-1}\cup \{b,a\},A_{v\in N_1}, k-1)$
    and in $F(j,k-1)$ inherits her approval ballot from $F(j-1,k-2)=(N_2, D_{k-2}\cup \{b,a\},A_{v\in N_2}, k-2)$, o again applying the inductive hypothesis to $j'=j-1$, we conclude that the claim holds for $v\in N_2$.
    We conclude that the claim holds for $j$ and $k>j+1$.
    
\end{proof}

We are now ready to present the omitted proofs.
\atomgroup*
\begin{proof}
First, we will argue that to prove the claim for $E(j,k)$ with $W=D_{k-2}\cup\{x,a\}$, it suffices to prove the claim for $F(j,k)$ with $W'= D_{k-1} \cup\{a\}$.

Indeed, recall that the set of voters in both elections is $N=N_1\sqcup N_2$ with $|N_1|=|N_2|=2^{j-1}$. 
Consider a voter $v\in N_1$. In $F(j, k)$ she inherits her ballot from $F(j-1, k)$, whereas in $E(j, k)$ she inherits her ballot from $F(j-1, k-1)$; in both cases, she modifies it by relabeling $a_1$ and $b_1$ as $b$ and $a$, and, in case of $E(j, k)$, she also adds $x$ to her ballot. Applying Proposition~\ref{prop:k-to-k-1} and observing that $D_{k-1}\subseteq W'$, $x\in W$, we conclude that the number of candidates in $W$ that $v$
approves in $E(j, k)$ equals to the number of candidates in $W'$ that she approves in $F(j, k)$; moreover, she approves the same subset of $\{a, b\}$ in both elections. 
On the other hand, voters in $N_2$ approve $y\notin W$ in $E(j,k)$, but not in $F(j,k)$, and, apart from $y$, they approve the same candidates in both elections.

By focusing on $F(j, k)$, which is defined inductively, 
we can prove both claims by induction on $j$; for readability, we denote the committee $D_{k-1}\cup\{a\}$ by $W$.\\
\noindent\textbf{Base case:} For $j=1$, we have $|A_1\cap W|=k$ and $|A_2\cap W|=k-2$, since voter $1$ approves all members of $W$, while voter $2$ approves everyone but $d_{k-1}$ and $a$.
Thus, the second claim of the lemma holds. For the first claim, observe that 
\begin{align*}
\Delta(W,a,b)=+\frac{1}{k-1}-\frac{1}{k}=\frac{1}{k(k-1)}=\delta(1,k).
\end{align*}
\noindent\textbf{Inductive Step:} Let $j\geq 1$, and suppose the statement of the lemma is true for $j$ and arbitrary $k>j$; we will show that it is true for $j+1$.
 Consider elections $F(j,k)$ and $F(j,k-1)$ with  $F(j,k)=(N_1,C_1,(A_v)_{v\in N_1},k)$, $F(j,k-1)=(N_2,C_2,(A_v)_{v\in N_2},k-1)$, and
$N_1\cap N_2= \varnothing$. Let $W_1=D_{k-1}\cup\{a_1\}$ and $W_2=D_{k-2}\cup\{a_2\}$. Applying the inductive hypothesis to $F(j,k)$ with committee $W_1$ and to $F(j,k-1)$ with committee $W_2$, we obtain
    	\begin{align*}
		& \Delta_{N_1}(W_1,a_1,b_1)=\delta(j,k)=\frac{j!}{\prod_{i=0}^{j}(k-i)}\text{ and }\\
		& \Delta_{N_2}(W_2,a_2,b_2)=\delta(j,k-1)=\frac{j!}{\prod_{i=0}^{j}(k-1-i)}.
	\end{align*}
Recall that the election $F(j+1,k)$ 
is constructed from elections $F(j,k)$ and $F(j,k-1)$ so that its set of voters is $N=N_1\sqcup N_2$ and candidates $a_1, b_1$ and $a_2, b_2$
are replaced with $b, a$ and $a, b$, respectively.
Thus, for $F(j+1,k)$ and  
$W=D_{k-1}\cup \{a\}$ we have
\begin{align}
		& \Delta_{N_1}(W,a,b)=\Delta_{N_1}(W_1\cup\{b_1\}\setminus\{a_1\},b_1,a_1)=-\Delta_{N_1}(W_1,a_1,b_1)=-\frac{j!}{\prod_{i=0}^{j}(k-i)}, \label{casek}\\
		& \Delta_{N_2}(W,a,b)=\Delta_{N_2}(W_2,a_2,b_2)=\frac{j!}{\prod_{i=0}^{j}(k-1-i)}.\label{casekm1}
\end{align}
 We conclude that in election $F(j+1,k)$ for committee $W$ we have
 \begin{align*}
 \Delta_{N}(W,a,b)
 &=\Delta_{N_1}(W,a,b)+\Delta_{N_2}(W,a,b)\\		
 & = -\frac{j!}{\prod_{i=0}^{j}(k-i)} + \frac{j!}{\prod_{i=0}^{j}(k-1-i)}=\frac{(k-(k-1-j))j!}{\prod_{i=0}^{j+1}(k-i)}\\
 &=\frac{(j+1)!}{\prod_{i=0}^{j+1}(k-i)}=\delta(j+1,k).
\end{align*}
  This proves the first part of the lemma. To prove the second part, consider first a voter $v\in N_1$
  and elections $F(j, k)$ and $F(j+1, k)$.  By the inductive hypothesis, $v$ approves at least
  $k-(j+1)$ candidates in $W_1$.
  As $v$ approves the same dummy candidates in both elections 
  and exactly one candidate from $\{a, b\}$, and
  $W=W_1\setminus\{a_1\}\cup\{a\}$, 
  it follows that $v$ approves at least $k-(j+1)-1=k-(j+2)$ candidates in $W$.
  Now, consider a voter $v\in N_2$. 
  By the inductive hypothesis, 
  in election $F(j,k-1)$ she 
  approves at least $k-1-(j+1)$
  candidates in $W_2$. In $F(j+1, k)$ she approves the same dummy candidates, and she approves $a$
  if and only if she approved $a_2$ in $W_2$.
   Since $W_2=D_{k-2}\cup\{a_2\}$, 
   $W=D_{k-1}\cup\{a\}$, 
   we conclude that $v$ approves at least  $k-1-(j+1)=k-(j+2)$
   candidates in $W$. We conclude that in election $F(j+1,k)$, for every $v\in N$ we have $|A_v\cap W|\geq k-(j+2)$, as desired.

\end{proof}
\goodseq*
\begin{proof}
Items (1) and (2) follow by applying Corollary~\ref{goodswap} for each $i\in[t]$.
Indeed, we have previously argued that the swaps $(c_{i},c_{i+1})$ or $(c_{i+1},c_i)$ only change the satisfaction of voters $N_i\subset N$. 
Furthermore none of $c_j, j\neq i,i+1$, are in the committee at any point of the sequence, so we may ignore them.
But then restricting $E^t(j,k)$ to $N_i$ and candidates $\{d_1,\ldots,d_{k-1},x,y,c_{i+1},x_{i+1}\}$ recovers an election isomorphic to $E(j,k)$ so that we can apply Corollary \ref{goodswap}
$t$ times to see that $\sum_{i=1}^{t} \Delta(W_1,c_{i},c_{i+1})= t \delta(j,k)$ and $\sum_{i=1}^{t} \Delta(W_2,c_{i},c_{i+1})= t \delta(j,k)$.

For items (3) and (4), note that by symmetry $\Delta(W_3,x,y)=\Delta(W_4,y,x)$ and furthermore 
\begin{align*}
	\sum_{i=1}^{t} \Delta(W_1,c_{i},c_{i+1})+\Delta(W_3,x,y)+\sum_{i=1}^{t} \Delta(W_2,c_{i},c_{i+1})+\Delta(W_4,y,x)=0,
\end{align*}
 as by executing this sequence of swaps we end up where we started, namely with committee $W_1$, implying that $\Delta(W_4,y,x)=\Delta(W_3,x,y)=-t\delta(j,k)$.
\end{proof}
\zerolow*
\begin{proof}
In this proof, we view committees as ordered $k$-tuples, as it will be useful to number the committee positions.
We will only consider committees of the form 
$$
(a_1,a_2,\ldots, a_{k_1},d_1,d_2,\ldots, d_{k_2}), 
\quad\text{where $a_i\in C_i\setminus D_{k_2}$, $1\leq i\leq k_1$.}
$$
We will refer to $a_i$ as the $C_i$-candidate of the committee, i.e., the candidate in committee position $i$, and say that the $C_i$-candidate has {\em odd} (resp. {\em even}) index if $i$ is odd (resp. even). Our super-polynomial swap sequence will consist of swaps of the form $(a,b)$, where $a,b\in C_i\setminus D_{k_2}$ for some  $i\in [k_1]$.
Given a committee $W$ of the form  above, we will say a voter group $N_i$, $i<k_1$, is
\textit{stable} if 
(1) $W$ contains $c_{i,1}$ and $c_{i+1,j}$ where $j$ is even or
(2) $W$ contains $c_{i,t+1}$ and $c_{i+1,j}$ and $j$ is odd.
We say that $N_{k_1}$ is {\em stable} if $c_{k_1,t+1}$ is on the committee.
By Proposition~\ref{goodseq}, if $N_i$ is stable then no swap $(a,b)$ with $a,b\in C_i$ is a good swap.
Let the initial committee be $$W_0=(c_{1,t+1},c_{2,t+1},\ldots,c_{k_1-1,t+1}, c_{k_1,1},d_1,\ldots d_{k_2})$$
so $|W_0|=k_1+k_2$. 
In $W_0$, all of the voter groups $N_i$, $i<k_1$, are stable, and $N_{k_1}$ is unstable. We will exhibit a sequence of good swaps that results in the final committee 
$$
(c_{1,t+1},c_{2,t+1},\ldots, c_{k_1,t+1},d_1,\ldots, d_{k_2}),
$$ 
where all voter groups $N_i$ are stable since $t+1$ is odd. We will call $N_{i-1}$ the \textit{predecessor} group of $N_i$. We will say a swap sequence {\em destabilises}  $N_{i}$ if prior to its execution it was stable and upon its execution $N_i$ is no longer stable.

To show that the sequence $\textbf{X}^1_{k_1}$ is a sequence of good swaps, we use induction on $i<k$ to show the following.
Suppose the groups of voters $N_1,\ldots N_{i-1}$ are stable and $N_{i}$ is the group of voters with the smallest index that is not stable (so if $i=1$, simply $N_1$ is not stable).
We want to prove inductively that in this case one of the swap sequences $\textbf{X}^1_i$ and $\textbf{X}^0_i$ is good.

To prove the claim for $i$, we split the argument into two cases. Consider committee $W$ of the form $$W=(c_{1,t+1},\ldots c_{i-1,t+1},c_{i,j},c_{i+1,j'},\ldots a_{k_1-1},a_{k_1}, d_1,\ldots,d_{k_2-1},d_{k_2})$$ where $j=1$ and $j'$ is odd
or $j=t+1$ and $j'$ is even. As before $a_\ell\in C_\ell\setminus D_{k_2}$, $k_1\geq \ell>i+1$. We show that in the former case, i.e., if the $C_{i+1}$-candidate in $W$ has odd index, then $\textbf{X}^1_i$ is a sequence of good swaps and in the latter case, i.e., if the $C_{i+1}$-candidate in $W$ has even index, then $\textbf{X}^0_i$ is a sequence of good swaps.

We prove the claim inductively on $i$. Suppose $N_1$ is not stable. If $c_{1,1}\in W$, then since only voters in $N_1$ approve candidates in $C_1$, 
by applying Proposition~\ref{goodseq} we conclude that the swap sequence $$\textbf{X}_1^{1}=(c_{1,1},c_{1,2}),\ldots,(c_{1,t}, c_{1,t+1})$$ is a sequence of good swaps, increasing the PAV score by $t\delta(2\lceil{\log k}\rceil,k_2+1)$. If $c_{1,t+1}\in W$, then by stability also $c_{2,j}\in W$ where $j$ is even.
In this case, also by Proposition \ref{goodseq}, $$\textbf{X}^0_1=(c_{1,t+1},c_{1,t}),\ldots, (c_{1,2},c_{1,1})$$ is a sequence of good swaps.

Now suppose $i>1$, voter groups $N_1,\ldots N_{i-1}$ are stable, 
and $N_i$ is unstable. We will consider two cases; in both cases we show that there is a good swap involving the $C_i$ candidate that destabilises the group $N_{i-1}$, allowing us to apply the inductive hypothesis.

Consider first the case
where $$W=(c_{1,t+1},\ldots c_{i-1,t+1},c_{i,j},c_{i+1,j'},\ldots, ,a_{k_1},d_1,\ldots,d_{k_2-1},d_{k_2})$$ where $j\leq t$ and both $j$ and $j'$ are odd.
We prove that $(c_{i,j},c_{i,j+1})$ is a good swap, i.e., it holds that $\Delta(W,c_{i,j},c_{i,j+1})>0$. From Lemma \ref{atomgroup}, it follows that 
$$
\Delta_{N_i}(W,c_{i,j},c_{i,j+1})=\delta(2\lceil{\log k}\rceil-2(i-1),k_2+1).
$$
Furthermore, by construction of $E$, none of the groups $N_j$, $j\notin\{i,i-1\}$, approve either of $c_{i,j},c_{i,j+1}$ so that $\Delta_{N_j}(W,c_{i,j},c_{i,j+1})=0$. So showing that $(c_{i,j},c_{i,j+1})$ is a good swap therefore boils down to showing that 
\begin{align*}
&\Delta_N(W,c_{i,j},c_{i,j+1})=\Delta_{N_{i-1}\sqcup N_i}(W,c_{i,j},c_{i,j+1})\\&=\delta(2\lceil{\log k}\rceil-2(i-1)),k_2+1)-t\delta(2\lceil{\log k}\rceil-2(i-2)),k_2+1)>0.
\end{align*}
We do this in the following lemma.

\begin{lemma}\label{gain}
$\delta(2\lceil{\log k}\rceil-2(i-1),k_2+1)>t\delta(2\lceil{\log k}\rceil-2(i-2),k_2+1)$.
\end{lemma}
\begin{proof}
By construction $2(i-1)\leq 2\lceil{\log k}\rceil$. Also, $t\leq k+1$ and for large enough $k$, we see that
	\begin{align*}
    &t\cdot\frac{(2\lceil{\log k}\rceil-2(i-2))!}{{\prod_{i=0}^{2\lceil{\log k}\rceil-2(i-2)}k_2+1-i}} < \frac{(2\lceil{\log k}\rceil-2(i-1))!}{\prod_{i=0}^{2\lceil{\log k}\rceil-2(i-1)}k_2+1-i}\\\iff&
	t\cdot (2\lceil{\log k}\rceil-2(i-2))(2\lceil{\log k}\rceil-2i+3)\\ &< (k_2+1-2\lceil{\log k}\rceil+2(i-2))(k_2+1-2\lceil{\log k}\rceil+2i-3),\end{align*}
 so since the leading term on the LHS is $O(k \log^2 k)$ as $t\leq k+1$ and the leading term on the RHS is $\Omega(k^2)$, the inequality holds for sufficiently large $k$.

\end{proof}
So $(c_{i,j},c_{i,j+1})$ is a good swap and
results in the committee 
$$
W=(c_{1,t},\ldots c_{i-1,t+1},c_{i,j+1},c_{i+1,j'},\ldots 
a_{k_1-1},a_{k_1},d_1,\ldots,d_{k_2-1},d_{k_2}), 
$$ 
which is not stable for $N_{i-1}$, as the $C_i$ candidate has even index $j+1$. By the inductive hypothesis, 
 $\textbf{X}^0_{i-1}$ is a sequence of good swaps that results in the committee 
 $$
 (c_{1,t+1},\ldots,c_{i-2,t+1}, c_{i-1,1},c_{i,j+1},c_{i+1,j'},\ldots a_{k_1-1},a_{k_1},d_1,\ldots,d_{k_2-1},d_{k_2}).
 $$
Now suppose the current committee is
$$
W=(c_{1,t+1},\ldots,c_{i-2,t+1}, c_{i-1,1},c_{i,j},c_{i+1,j'},\ldots a_{k_1-1},a_{k_1},d_1,\ldots,d_{k_2-1},d_{k_2}),
$$ 
where $j\leq t$ is even and $j'$ is odd. The argument that $\Delta_N(W,c_{i,j},c_{i,j+1})>0$ is the same as before: 
\begin{align*}
\Delta_N(W,c_{i,j},c_{i,j+1})
&=\Delta_{N_{i-1}\cup N_i}(W,c_{i,j},c_{i,j+1})\\
&=\delta(2\lceil{\log k}\rceil-2(i-1)),k_2+1)-t\delta(\log k-2(i-2)),k_2+1)>0
\end{align*}
by Lemma \ref{gain}. So we move to the committee 
$$
W=(c_{1,t+1},\ldots,c_{i-2,t+1}, c_{i-1,1},c_{i,j+1},c_{i+1,j'},\ldots a_{k_1-1},a_{k_1},d_1,\ldots,d_{k_2-1},d_{k_2}).
$$
However, in this case $j+1$ is odd, and so by the inductive hypothesis $\textbf{X}^1_{i-1}$ is a sequence of good swaps resulting in the committee 
$$
(c_{1,t+1},\ldots c_{i-1,t+1},c_{i,j+1},c_{i+1,j'},\ldots a_{k_1-1},a_{k_1},d_1,\ldots,d_{k_2-1},d_{k_2}).
$$
Note that in the resulting committees $N_i$ is unstable, and we can keep repeating our argument until the $C_{i}$ candidate has index $t$.
This proves that if $j=1$ (and $j'$ is odd), then the sequence of swaps 
$$
\textbf{X}_i^1=
\oplus_{j=1}^{t}\left((c_{i,j},c_{i,j+1})\oplus\textbf{X}_{i-1}^{j-1 \mod 2}\right)
$$
is a sequence of good swaps resulting in the committee $$W=(c_{1,t+1},c_{2,t+1},\ldots, c_{i,t+1},c_{i+1,j'},\ldots,d_1,\ldots, d_{k_2-1},d_{k_2}).$$

\noindent The case where $j'$ is even can be shown similarly. In particular, consider the committee $$W=(c_{1,t+1},\ldots c_{i-1,t+1},c_{i,j},c_{i+1,j'},\ldots, ,a_{k_1}, d_1,\ldots,d_{k_2-1},d_{k_2})$$ where $j>0$ is odd and $j'$ is even. Again, the swap $(c_{i,j},c_{i,j-1})$ is a good swap and destabilises the voters $N_{i-1}$. Indeed, by construction only voters in $N_i$ and $N_{i-1}$ 
can distinguish $c_{i,j}$ and $c_{i,j-1}$, so
\begin{align*}\Delta(W,c_{i,j},c_{i,j-1})=\Delta_{N_{i-1}\sqcup N_i}(W,c_{i,j},c_{i,j-1})\end{align*}
and so we can apply Proposition \ref{goodseq} and Corollary \ref{goodswap} to the induced instances together with Lemma \ref{gain} to conclude that $\Delta(W,c_{i,j},c_{i,j-1})=\delta(2\lceil{\log k}\rceil-2(i-1)),k_2+1)-t\delta(2\lceil{\log k}\rceil-2(i-2)),k_2+1)>0$.
Then, by the inductive hypothesis the sequence of swaps $\textbf{X}^1_{i-1}$ is a good sequence of swaps.\\
If instead, the committee initially is $$W=(c_{1,t+1},\ldots c_{i-1,1},c_{i,j},c_{i+1,j'},\ldots, ,a_{k_1},d_1,\ldots,d_{k_2-1},d_{k_1})$$ where $j>0$ is even and $j'$ are even, then by an analogous argument the swap $(c_{i,j},c_{i,j-1})$ is good and destabilises $N_{i-1}$, so by the inductive hypothesis $\textbf{X}^0_{i-1}$ is a sequence of good swaps.

Finally, note that Lemma~\ref{gain} also applies to $i=k_1$, meaning that if we have a committee $W$ with $c_{k_1,j}\in W, j\leq t$, which is stable for $N_1,\ldots, N_{k_1-1}$, the swap $(c_{k_1,j},c_{k_1,j+1})$ is good.
If $j$ is even, we just proved that the sequence of swaps $\textbf{X}_{k_1-1}^1$ is a sequence of good swaps and if $j$ is odd, then by the inductive hypothesis $\textbf{X}_{k_1-1}^0$ is a sequence of good swaps. 
This proves that indeed $\textbf{X}_{k_1}^1$ is a sequence of good swaps for initial committee $W_0$ and we are done.
\end{proof}
\smallskip
\lexresponse*
\smallskip
\begin{proof}
Consider election $E$ from the proof of Theorem \ref{thm:zerolow} and the linear order $<$ on the set of candidates that satisfies 
\begin{itemize}
\item $c_{i,j}<c_{i,j+1}$,
\item $c_{i,j}<c_{i+1,j'}$ for all $j,j'$ and 
\item $c_{i,j}<d_1<\ldots<d_{k_2}$ for all $i,j$.
\end{itemize}
We will slightly modify $E$ by adding $\mathrm{poly}(k)$ voters. In particular, these voters cannot distinguish between the candidates within a set $C_i$, so that any swaps involving two candidates in $C_i$ are unaffected by these new voters. Their purpose is to prevent swaps $(a,b)$
 where $a\in C_{i}\cup D_{k_2}, b\in C_j$, $j\neq i$ from being executed by $0^+$-ls-PAV with the lexicographic pivoting rule.
Suppose we have a bound $\gamma$ such that at any point of the swap sequence no swap can increase the PAV score by more than $\gamma$. That is, for committee $W$ resulting from the execution of a prefix of the swap sequence $\textbf{X}^1_{k_1}$, it holds that whenever $c\in W, c'\notin W$, then  $\Delta(W,c,c')\leq \gamma.$
Assuming for now we have such a bound, we add voters disjoint sets of voters $V_i$, $i\in [k]$ with $|V_i|= \lceil{2\gamma}\rceil$. For $1\leq i \leq k_1$, the voters in $V_i$ approve of all candidates in $C_i$ , while voters in $V_i, i>k_1$, approve $d_{i-k_1}$.
Let $c_i\in C_i, c_j\in C_j$ and $i\neq j$ and assume $c_i\in W$ where $W$ is the result of the execution of some prefix of the swap sequence $\textbf{X}^1_{k_1}$.
This means that $W$ contains exactly one candidate from $C_i$, namely $c_i$, and exactly one candidate from $C_j$ but not $c_j$.
We observe that $$\Delta_{V_i\cup V_j}(W, c_i, c_j)= -\lceil{2\gamma}\rceil +\frac{1}{2} \lceil{2\gamma}\rceil=-\frac{\lceil{2\gamma}\rceil}{2}\leq -\gamma,$$
so \begin{align*}\Delta_{N\cup_{l=1}^{k} V_l}(W,c_i,c_j)=\Delta_{N}(W,c_i,c_j)+\Delta_{V_i\cup V_j}(W,c_i,c_j)\leq \gamma -\gamma =0,\end{align*}
implying that (as a result of the presence of the additional voters) the swap $(c_i,c_j)$ is not a good swap.
Similarly for $i>k_1$, $$\Delta_{N \cup_{l=1}^{k}V_i}(W,d_i,c_j)=\Delta_{N}(W,d_i,c_j)+\Delta_{V_i\cup V_j}(W,d_i,c_j)\leq \gamma-\gamma =0,$$ so $(d_i,c_j)$ is not a good swap for $W$.
Since we only need a bound $\gamma$ that is $\mathrm{poly}(k)$, for simplicity we can use the trivial upper bound that the addition of a candidate $c_j$ improves the contribution of its supporters to the PAV score by at most $1$ in $E$. Since $c_j$ is supported by a subset of voters in $N_j$ and if $j>1$ also $N_{j-1}$, a trivial bound is 
$\Delta(W,c_i,c_j)\leq \Delta(W,c_j)\leq 2|N_j| =2\cdot 4 k^2(k+1)=8k^2(k+1)$ (for sufficiently large $k$) and similarly $\Delta(W,d_i,c_j)\leq 8k^2(k+1)$, $i\in [k_2]$, so $\gamma=8k^2(k+1)$.\\\\
We will now argue that $0^+$-ls-PAV will execute a swap sequence of superpolynomial length on this modified instance. We claim $0^+$-ls-PAV will not execute the same swaps $\textbf{X}^1_{k_1}$ but instead \textit{shortcuts} this sequence of swaps. What we mean by this is that some subsequences of the form $(a_1,a_2)(a_2,a_3),\ldots, (a_{l-2},a_{l-1}),(a_{l-1},a_l)$ will be replaced by the new subsequence $(a_1,a_l)$.
We will write 
$$
(a_1,a_2)(a_2,a_3),\ldots, (a_{l-2},a_{l-1}),(a_{l-1},a_l)\xrightarrow{\text{replace}}(a_1,a_l).
$$
We will argue that this type of shortcutting does not occur too much, so that even with the lexicographic pivoting rule the number of iterations remains superpolynomial. 
Moreover, the intermediate committees that result from any prefix of the swap sequence thus obtained are a subset of the intermediate committees obtained in $\textbf{X}^1_{k_1}$.\\To obtain the new sequence of swaps we shortcut subsequences in $\textbf{X}^1_{k_1}$ as follows: \begin{center} We iterate from $l=k_1-1$ to $l=1$ and for a subsequence $\textbf{X}^0_l$ of $\textbf{X}^1_{k_1}$ update the sequence via $$\textbf{X}^0_l\xrightarrow{\text{replace}}(c_{l,t+1},c_{l,1})$$\end{center}
Proceeding this way, once we replace a subsequence $\textbf{R}=\textbf{X}^0_i$ with the single swap $(c_{i,t+1},c_{i,1})$, $\textbf{R}$ is not the subsequence of a sequence of the form $\textbf{X}^0_l$, $l>i$, as by design of the procedure $\textbf{X}^0_l$ was already replaced earlier in the process.
 We call the sequence obtained from $\textbf{X}^1_i$ by replacing sequences in this manner $\textbf{Z}^1_i$. Formally, 
 \begin{align*}
 \textbf{Z}^1_{i}=\oplus_{j=1}^{t}\left((c_{i,j},c_{i,j+1})\oplus\textbf{Z}_{i-1}^{j-1 \mod 2}\right)\text{ for }i>1
 \end{align*}
 where $\textbf{Z}_{i-1}^{0}=(c_{i,t+1},c_{i,1})$ and $\textbf{Z}_1^1=\textbf{X}^1_1$.
To conclude our argument, we need to show the following two items.
\begin{enumerate}[{(1)}]
\item $0^+$-ls-PAV executes $\textbf{Z}^1_{k_1}$ on $E$ when using better response with the fixed order $<$  and\label{item1}
\item $\textbf{Z}^1_{k_1}$ has superpolynomial length.\label{item2}
\end{enumerate}

For \Cref{item1}, note first that all swaps in $\textbf{Z}^1_{k_1}$ are of the form $(a,b)$ with $a,b\in C_i$ for some $i \in [k_1]$. Indeed, we previously showed that for any committee $W$ resulting from executing a prefix of $\textbf{X}^1_{k_1}$ any swaps of the form \begin{itemize}
\item 
$(c_{i},c_j)$, $c_i\in C_i$ and $c_j\in C_j, i\neq j$
\item 
$(d_i,c_j)$, $c_j\in C_j, i\neq j$,
\end{itemize}
do not increase the PAV score. Moreover, any committee $W$ that results from executing a prefix of $\textbf{Z}^1_{k_1}$ is also the resulting committee of some prefix of $\textbf{X}^1_{k_1}$,
and hence this also holds for our sequence $\textbf{Z}^1_{k_1}$.
It remains to prove that for $W$ resulting from the execution of some prefix of $\textbf{Z}^1_{k_1}$, the following holds: If the next swap is $(a,b), a,b\in C_i$, then no other swap $(c,d)$ that lexicographically precedes $(a,b)$ is a good swap. With this in mind, observe that for every swap $(a,b)$ in $\textbf{X}^1_{k_1}$ with $a,b\in C_i$ it holds that the committee $W$ resulting from the prefix in $\textbf{X}^1_{k_1}$ preceding $(a,b)$ is stable for any $N_j, j<i$. This stability implies that no swap $(c,d)$, $c,d\in C_j$, $j<i$, is good. $\textbf{Z}^1_{k_1}$ inherits this property since it inherits the intermediate committees from $\textbf{X}^1_{k_1}$. Furthermore, by construction of $\textbf{Z}^1_{k_1}$, $(a,b)$
is either $(c_{i,j},c_{i,j+1})$ for some $j\leq t$ or $(c_{i,t+1},c_{i,1})$.
By the above, the latter is clearly lexicographically minimal among good swaps. For the former observe that
since $(c_{i,j},c_{i,j+1})$ is good, this implies that $N_i$ is unstable with $i$ minimal, so either $i=k_1$ or else $i<k_1$ and $c_{i+1,p}\in W$ where $p$ is odd. So any swap $(c_{i,j},c_{i,l})$, $l<j$, is bad. This concludes the proof of \Cref{item1}.

\smallskip
To prove \Cref{item2}, we now show that $\textbf{X}^1_{k_1}$ still has length $\Omega(k^{\log k})$. Consider $\textbf{X}^1_{2}$ (which is not the subsequence of any $\textbf{X}^0_l$, $l>2$). 
Under the lexicographic pivoting rule, the sequence $\textbf{X}^0_1$ following a swap $(c_{2,j},c_{2,j+1})$ where $j$ is odd gets replaced with the single swap $(c_{1,t+1},c_{1,1})$, but at least half the swaps remain intact. So $|\textbf{Z}^1_2|\geq \frac{1}{2} |\textbf{X}^1_2|\geq \frac{t^2}{2}$. Similarly, $|\textbf{Z}^1_{i+1}|=|\oplus_{j=1}^{t}(c_{i,j},c_{i,j+1})\textbf{Z}_{i}^{j-1 \mod 2}|\geq \frac{t}{2}|\textbf{Z}^1_{i}|$ (remember that $t$ is even) so that $|\textbf{Z}^1_{k_1}|\geq \frac{t}{2}^{k_1-2}|\textbf{Z}^1_2|\geq \frac{t}{2}^{k_1-2}\frac{t^2}{2}=\Omega(\frac{t}{2}^{k_1})=\Omega(k^{k_1})$ since $t\geq k$.
 \end{proof}
\newpage

\section{Better Response vs Best Response: Empirical Analysis}\label{sec:exp}
 We have seen that, in the worst case, $0^+$-ls-PAV with a lexicographic pivoting rule may require a number of iterations that is super-polynomial in the committee size.
 Another natural pivoting rule for $0^+$-ls-PAV is to always consider all possible
 swaps and perform the swap that offers the maximum increase in the PAV score, i.e., 
 to pick a best response rather than simply a better response.
Intuitively, best response short-cuts the search by avoiding low-value swaps, 
and may therefore require fewer iterations. Indeed, we were unable to extend
out lower bound construction to the best-response version of $0^+$-ls-PAV, 
so it remains an open problem whether best response may take a superpolynomial number
of steps to converge.

On the other hand, best response necessarily considers $k(m-k)$ swaps in each iteration, while better response may be able to identify a sufficiently good swap after considering just a few candidate pairs. Thus, it is not clear which of these pivoting rules should be preferred in practice. In this section, we provide evidence from both real-world and synthetic data sets that supports the use of better response.

\subsection{Simulation Set-up}
We take a computational/data-driven approach and compare better and best response on data sets available on PrefLib\footnote{https://www.preflib.org/} as well as on synthetic data sets. As a proxy for running time, we consider the number of \textit{comparisons}---the number of evaluations of the quantity $\Delta(W, c, c')$---during the execution of each algorithm. For simplicity, in all our experiments, we choose a committee with the maximum number of approvals (i.e., $W\in\arg\max\sum_{v\in N}|A_v\cap W|$) as our starting point.
We begin by describing our two PrefLib data sets and subsequently move to describing the random approval sample models we use in our simulations.\\\\

\noindent\textbf{AAMAS 2015 / 2016 Bidding Data}
This pair of datasets contains the bids of reviewers over papers from the 2015 and 2016 editions of Autonomous Agents and Multi-agent Systems Conference. Inclusion in these datasets was explicitly opt-in. The 2015 dataset contains 9,817 bids of $201$ reviewers on $613$ papers; this represents about $40\%$ of the actual $22,360$ bids of $281$ reviewers over $670$ papers. The 2016 dataset contains $161$ out of $393$ reviewers with bids over $442$ out of $550$ papers.
The bidding language for these conferences consisted of four options for each paper: `yes', `maybe', `no', and `conflict'. 
We merge answers categories `yes' and `maybe' to get an approval vote, so a voter approves a paper if and only if she selected `yes' or `maybe'. The committee size is not part of the input, so we vary it from 1 to 10 or 30, depending on the instance size. The results of our simulations are displayed in Figure \ref{fig:aamas}.\\\\
\begin{figure*}[t!]
\centering
\begin{subfigure}{0.49\textwidth}
  \includegraphics[width=\linewidth]{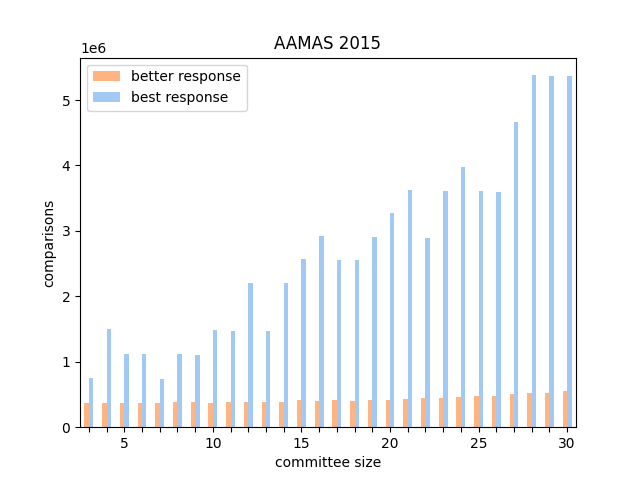}
   \caption{Number of swaps ls-PAV considers\\ on AAMAS 2015 bidding data}
  \label{fig:comps1}
\end{subfigure}
\hfill
\begin{subfigure}{0.49\textwidth}
  \includegraphics[width=\linewidth]{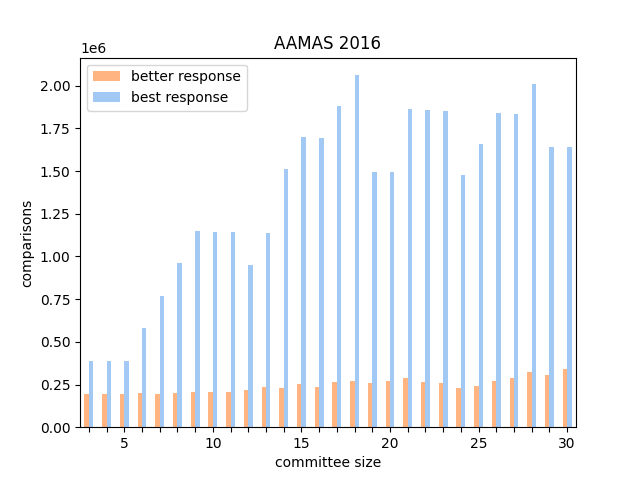}
 \caption{Number of swaps ls-PAV considers\\ on AAMAS 2016 bidding data}
  \label{fig:comps2}
\end{subfigure}

\caption{$0^+$-ls-PAV with best response vs.~$0^+$-ls-PAV with better response. On the x-axis, we vary the committee size, and on the y-axis we display the number of swaps considered. }\label{fig:aamas}
\end{figure*}

\noindent\textbf{Synthetic Data}
We generate synthetic data using three approval sampling models: $p$-Impartial culture, the $p$-$\phi$-resampling model \cite{SF+22} and the $d$-Euclidean model \cite{BL07}, as discussed in the following. In each case we generate elections with $100$ voters, $20$ candidates and a committee of size $k\in\{3,4,\ldots,10\}$, repeat this $10000$ times and display our results in a boxplot. We collect the results 
in Figure~5 (one parameter setting for each model); Figure 6 provides additional results for other parameter settings.

\smallskip

\noindent\textbf{$p$-Impartial culture} Each voter approves each candidate independently with probability $p$. In our experiments we use $p\in\{0.25,0.5,0.75\}$. The result for $p=0.5$ is displayed in Figure \ref{fig:synth1}, while the other cases can be found in Figure \ref{fig:syntheticapp}.

\smallskip

\noindent\textbf{$p$-$\phi$-Resampling}
For given $p,\phi \in [0,1]$, uniformly sample a central voter $u$ with $\lfloor{pm}\rfloor$ approvals.
We generate each new vote $v$ by initially setting $A_v=A_u$. We then execute the following procedure for every candidate $c \in C$: With probability $1-\phi$, we leave $c$'s approval intact. With the remaining probability $\phi$ we resample its value, by having $v$ approve $c$ with probability $p$ and not approve $c$ with probability $1-p$. 
We run experiments for the parameter settings $p=0.2,\phi=0.5$ (Figure \ref{fig:synth4}) and $p=0.5,\phi=0.5$ (Figure \ref{fig:app6}).

\smallskip

\noindent\textbf{$d$-Euclidean domain}
Each voter $v$ and candidate $c$ are associated with a position $p_v$ and $p_c$ in the $d$-dimensional unit cube $[0,1]^d$. We consider $d\in\{1,2\}$. 

We denote by $r>0$ the approval radius.
A voter at position $p_v$ approves a candidate at position $p_c$ if the Euclidean distance between them is less than $r$, i.e. $d_1(p_v,p_c)<r$ ($d_2(p_v,p_c)<r$ respectively).
For the $d$-dimensional Euclidean domain, we sample positions of voters and candidates uniformly from $[0,1]^d$.
Our parameter settings are as follows. For $d=1$ we consider $r=0.01$ (Figure \ref{fig:synth2}) and $r=0.1$ (Figure \ref{fig:app4}). For $d=2$ we consider $r=0.1$ (Figure \ref{fig:synth3}).

\begin{figure*}[t!]
\centering
  
\begin{subfigure}{0.49\textwidth}
  \includegraphics[width=\linewidth]{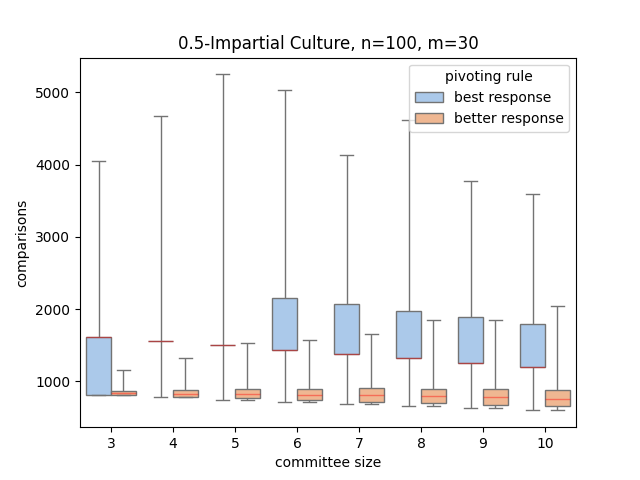}
   \caption{$p$-impartial culture for $p=0.5$}
  \label{fig:synth1}
\end{subfigure}
\hfill
\begin{subfigure}{0.49\textwidth}
  \includegraphics[width=\linewidth]{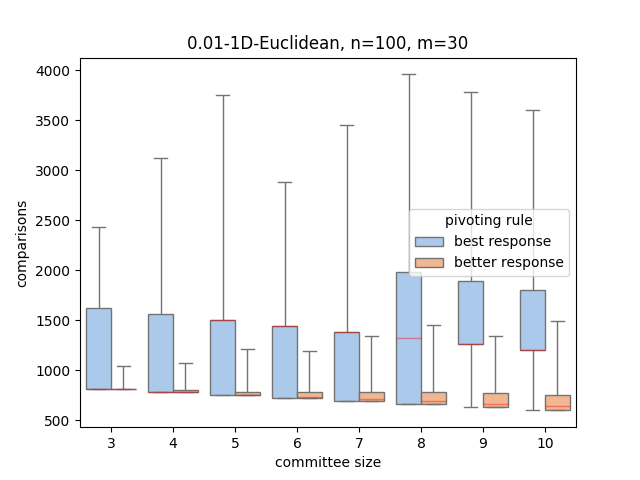}
 \caption{$1$-D Euclidean with radius $r=0.01$}
  \label{fig:synth2}
  
\end{subfigure}
\hfill
\begin{subfigure}{0.49\textwidth}
  \includegraphics[width=\linewidth]{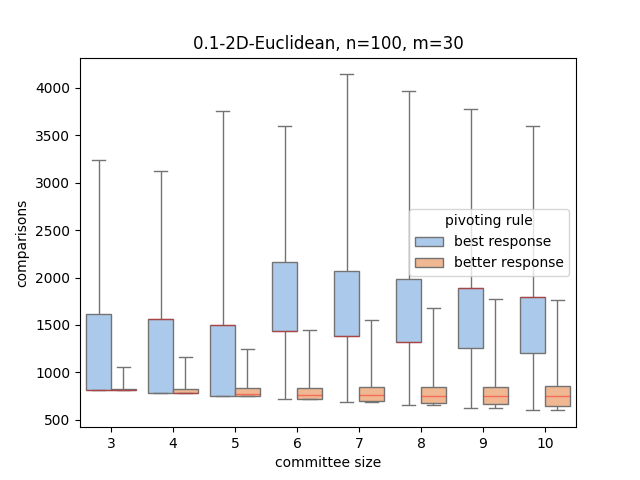}
 \caption{$2$-D Euclidean with radius $r=0.1$}
  \label{fig:synth3}
  
\end{subfigure}
\hfill
\begin{subfigure}{0.49\textwidth}
  \includegraphics[width=\linewidth]{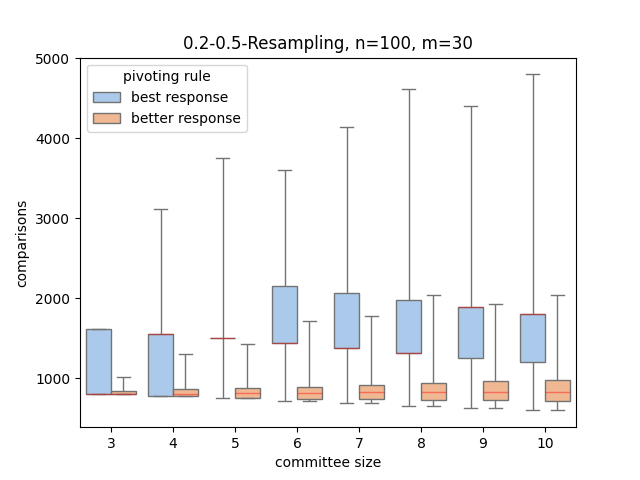}
 \caption{$p$-$\phi$-Resampling for $p=0.2$ and $\phi=0.5$}
  \label{fig:synth4}
  
\end{subfigure}
\caption{$0^+$-ls-PAV with best response vs.~$0^+$-ls-PAV with better response. On the x-axis, we vary the committee size. On the y-axis we display median, maximum, minimum, 25th and 75th quantile of the number of comparisons over 10000 repetitions on the described random model with $n=100$ voters and $m=30$ candidates in standard boxplot format. }\label{fig:synthetic}
\end{figure*}
\subsection{Takeaway}
Across our experiments we observe that better response outperforms best response. As we see in Figure~\ref{fig:comps1}, on the AAMAS 2015 data set, the number of swaps considered by the best response is up to $6$ times larger than the number of comparisons made by the better response. Further, as we increase the committee size from $3$ up to $30$, for better response the increase (if any) in the number of swaps considered is very slow.
In contrast, for the best response the number of comparisons increases over $5$-fold in both AAMAS datasets. 
For our synthetic data sets we observe the same trend, although the difference between better and best response does not appear to be quite as stark due to smaller election sizes. As can be observed in Figures~\ref{fig:synthetic} and~\ref{fig:syntheticapp}, across our simulations best response requires nearly twice as many comparisons half of the time and the maximum number of comparisons observed is usually as least twice as large compared as to better response. We note that various other parameter settings we tested followed the same pattern.\begin{figure}[h]
\centering
\begin{subfigure}{0.49\textwidth}
  \includegraphics[width=\linewidth]{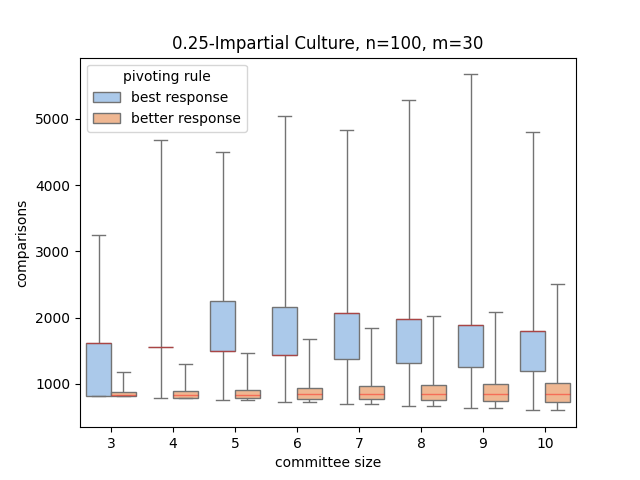}
 \caption{$p$-impartial culture for $p=0.25$}
  \label{fig:app1}
\end{subfigure}
\hfill
\begin{subfigure}{0.49\textwidth}
  \includegraphics[width=\linewidth]{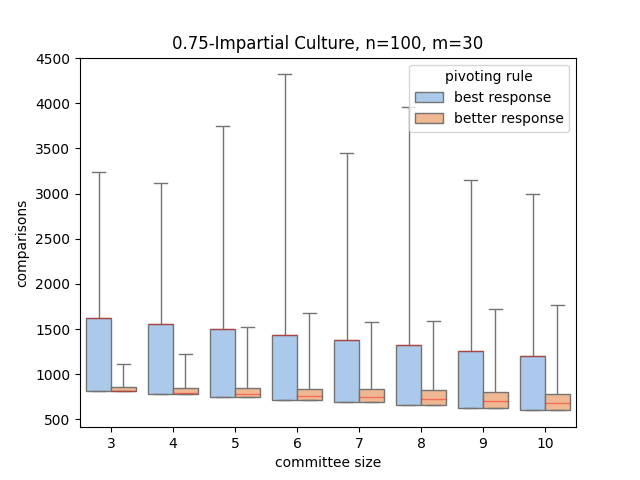}
 \caption{$p$-impartial culture for $p=0.75$}
  \label{fig:app2}
  
\end{subfigure}

\hfill
\begin{subfigure}{0.49\textwidth}
  \includegraphics[width=\linewidth]{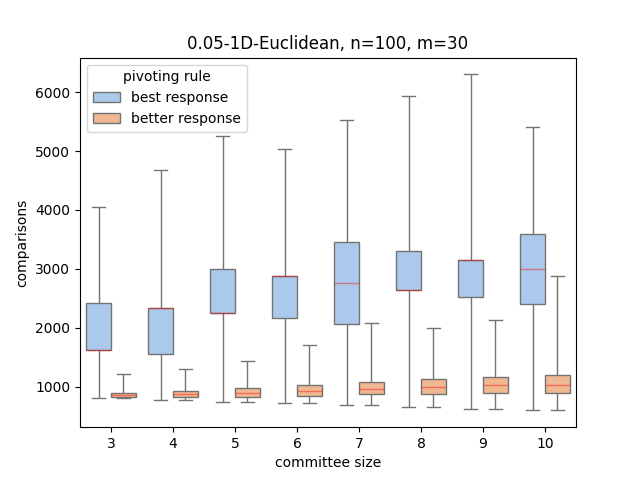}
 \caption{$1$-D Euclidean with radius $r=0.05$}
  \label{fig:app4}
  
\end{subfigure}
\hfill
\begin{subfigure}{0.49\textwidth}
  \includegraphics[width=\linewidth]{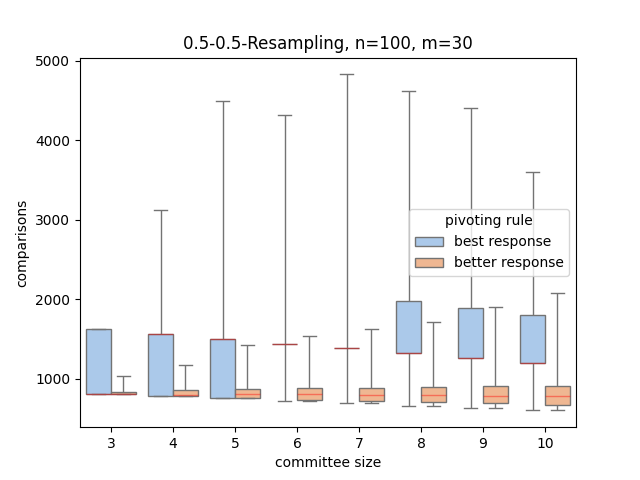}
 \caption{$p$-$\phi$-Resampling for $p=0.5$ and $\phi=0.5$}
  \label{fig:app6}

\end{subfigure}
  
\caption{$0^+$-ls-PAV with best response vs.~$0^+$-ls-PAV with better response. On the x-axis, we vary the committee size. On the y-axis we display median, maximum, minimum, 25th and 75th quantile of the number of comparisons over 10000 repetitions on the described random model with $n=100$ voters and $m=30$ candidates in standard boxplot format. }\label{fig:syntheticapp}
\end{figure}

\end{document}